\documentclass[11pt]{amsart}
\usepackage{fullpage}
\usepackage{graphicx} 
\usepackage{aliascnt}
\usepackage{amsmath,amssymb,amsfonts}

\newcommand{\SU}{{\mathsf{SU}}}
\newcommand{\PU}{{\mathsf{PU}}}
\newcommand{\LU}{{\mathsf{LU}}}
\newcommand{\UU}{{\mathsf{U}}}
\newcommand{\SO}{{\mathsf{SO}}}
\newcommand{\OO}{{\mathsf{O}}}
\newcommand{\Sp}{{\mathsf{Sp}}}
\newcommand{\Cl}{{\mathsf{Cl}}}
\newcommand{\LCl}{{\mathsf{LCl}}}

\newcommand{\su}{{\mathfrak{su}}}
\newcommand{\so}{{\mathfrak{so}}}

\newcommand{\symgp}{{\mathsf{S}}}
\newcommand{\Ind}{{\mathsf{Ind}}}
\newcommand{\Res}{{\mathsf{Res}}}
\newcommand{\id}{{\mathsf{id}}}

\usepackage{xcolor}
\definecolor{darkblue}{rgb}{0.,0.,0.4}
\definecolor{darkred}{rgb}{0.5,0.,0.}
\definecolor{darkorange}{rgb}{1.0, 0.55, 0.0}
\usepackage[pdftex,colorlinks=true,linkcolor=blue,citecolor=darkred,urlcolor=darkblue]{hyperref}

\usepackage{amsmath,amsthm,amsfonts,amssymb,braket}
\usepackage{fullpage}
\usepackage{mathtools}
\usepackage[numbers]{natbib}
\usepackage[normalem]{ulem}
\usepackage{booktabs}

\usepackage{tikz-cd}
\usetikzlibrary{arrows.meta,decorations.pathreplacing,decorations.pathmorphing,calc}

\numberwithin{equation}{section}

\newtheorem{lemma}[equation]{Lemma}
\newtheorem{theorem}[equation]{Theorem}
\newtheorem{conjecture}[equation]{Conjecture}
\newtheorem{corollary}[equation]{Corollary}
\newtheorem{proposition}[equation]{Proposition}
\newtheorem{proposition-definition}[equation]{Proposition-Definition}
\theoremstyle{definition}

\newtheorem{noteadded}[equation]{Note added}
\newtheorem{definition}[equation]{Definition}
\newtheorem{remark}[equation]{Remark}

\newcommand{\CC}{{\mathbb{C}}}
\newcommand{\RR}{{\mathbb{R}}}
\newcommand{\ZZ}{{\mathbb{Z}}}
\newcommand{\FF}{{\mathbb{F}}}

\newcommand{\calG}{{\mathcal{G}}}
\newcommand{\calH}{{\mathcal{H}}}

\newcommand{\calO}{{\mathcal{O}}}
\newcommand{\calP}{{\mathcal{P}}}
\newcommand{\calQ}{{\mathcal{Q}}}

\newcommand{\calS}{{\mathcal{S}}}

\newcommand{\calZ}{{\mathcal{Z}}}

\newcommand{\one}{{\mathbf{1}}}
\newcommand{\rd}{{\mathrm{d}}}

\newcommand{\Asterisk}{{\scalebox{1.7}{\raisebox{-0.2ex}{$\ast$}}}}
\DeclareMathOperator{\convo}{\Asterisk}

\newcommand{\ii}{{\mathrm{i}}}

\DeclareMathOperator{\Tr}{{Tr}}

\DeclareMathOperator*{\diag}{{diag}}

\DeclareMathOperator{\End}{{End}}

\DeclareMathOperator{\Ex}{\mathbb{E}}
\DeclareMathOperator*{\avg}{\Ex}

\DeclarePairedDelimiter{\norm}{\lVert}{\rVert}
\DeclarePairedDelimiter{\abs}{|}{|}

\title{Random unitary circuits with constant spectral gap}
\author{Tim Baer$^*$ and Jeongwan Haah$^{*\dagger}$\\ \\
$^*$Stanford University, Stanford, California, USA\\
$^\dagger$Google Quantum AI, California, USA}

\begin{document}
\maketitle
\begin{abstract}
    We prove constant lower bounds for the spectral gap of the following random walks 
    on unitary groups~$\mathsf{SU}(2^n)$ on $n$ qubits.
    (i) Random Pauli Rotation:
    choose an $n$-qubit Pauli operator~$P$ 
    and an angle~$\theta \in [0,2\pi)$, both uniformly at random,
    and apply $e^{\mathrm i \theta P}$.
    (ii) Brickwork Random Unitary Circuit:
    choose $n-1$ unitaries $U_{i}$ uniformly at random from~$\mathsf{SU}(4)$ (including central phases of the parent group) independently,
    and apply $U_{2j-1}$ on two qubits~$2j-1, 2j$ and then 
    $U_{2j}$ on two qubits~$2j, 2j+1$.
    Importantly, the gap bounds are independent of~$n$ and apply
    for all finite dimensional unitary representations of~$\mathsf{SU}(2^n)$ uniformly.

    Our results immediately imply that the Random Pauli Rotation and Brickwork Random Unitary Circuit converge to approximate unitary $t$-designs after $\mathcal{O}(nt)$ steps.
    As a consequence, we improve implications for
    the complexity growth of random unitary circuits as a function of the depth.
    The existing argument can be applied almost verbatim,
    but previously the depth had to be~$\mathcal{O}(2^{cn})$ for some~$c < 1$.
    Our unconditional constant spectral gap removes this restriction.
    
    We prove analogous constant-gap results
    for orthogonal and Clifford groups;
    the latter is indispensable
    for our result on the Brickwork Random Unitary Circuit.
    The former implies convergence to approximate orthogonal $t$-designs.
    In this case, the implied mixing time with respect
    to a constant additive and multiplicative error
    is optimal (up to an absolute constant factor) for~$t=2$.
\end{abstract}

\clearpage
\tableofcontents

\section{Introduction}

Random walks on groups are studied in quantum information theory, computer science, and probability
as they provide efficient ways to generate (approximately) Haar random distributions.
For example, 
brickwork random unitary circuits~\cite{brandao2016local,
Nahum_2017,
Nahum_2018,
brandao2021models,brown2018second,Fisher_2023} 
are among the simplest ensembles of random walks
that lie at the intersection of quantum many-body physics, quantum chaos, 
and quantum computing.
We are concerned with the rate of convergence of random walks on unitary groups
measured by the following.

\begin{definition}\label{defn:spectralgap}
    Let $G$ be a finite or compact Lie group 
    with the Haar probability measure~$\mu(G)$, 
    $\nu$~be a probability measure on~$G$,
    and $\rho$ be a finite dimensional unitary representation of $G$.
    We define
    \begin{align*}
        &\text{moment operator} &M(\nu,\rho,G) &= \avg_{g\sim \nu} \rho(g),\\
        &\text{essential norm} &g(\nu,\rho, G) &= \norm{M(\nu,\rho,G) - M(\mu(G), \rho, G)}_\infty,\\
        &\text{spectral gap of $\nu$ in $\rho$} & \Delta(\nu,\rho,G) &= 1 - g(\nu,\rho, G),\\
        &\text{spectral gap of $\nu$} & \Delta(\nu,G) &= \inf_\rho \Delta(\nu,\rho, G),
    \end{align*}
where the norm is the operator norm and $\rho$ in the infimum 
ranges over all finite dimensional unitary representations of~$G$.
When the group $G$ is evident, we may omit it from the moment operator notation.
\end{definition}
In this paper we consider~$\SU(2^n)$ and its subgroups.
This Lie group is identified with the unitary group 
acting on $n$ qubits~$(\CC^{2})^{\otimes n}$.
The random walks we study in this paper are the following.
\begin{definition}
    Let $\calP_n = \{\one_2, \sigma^x, \sigma^y, \sigma^z\}^{\otimes n} \setminus \{\one_2^{\otimes n}\}$
    be a set of $n$-qubit non-identity Pauli operators.
    The (special) Clifford group denoted by~$\Cl(n)$ 
    is the normalizer in~$\SU(2^n)$ 
    of the (special) Pauli group~$\langle \calP_n, \ii \one_{2^n} \rangle \cap \SU(2^n)$.%
    \footnote{A conventional Pauli group contains elements with determinant different from~$1$.
    We take only those with determinant~$1$.}
    For a subset~$S \subseteq \{1,2,\ldots,n\}$,
    we define local groups $\LCl(S,n) \subseteq \Cl(n)$ and $\LU(S,n) \subseteq \SU(2^n)$
    to be the subgroup consisting of all elements that are supported on~$S$;
    the center of the parent group is
    the smallest local group~$\LCl(\emptyset,n)$ or~$\LU(\emptyset,n)$,
    which are discrete.
    We will often write $\LU(m)$ and $\LCl(m)$ where $m = \abs{S}$ in place of~$\LU(S,n)$ 
    and $\LCl(S,n)$, respectively,
    when $S$ and $n$ are clear from the context.
    A step in our random walk is as follows,
    which specifies a probability distribution~$\nu$ on a designated compact group.
    \begin{itemize}
        
        \vspace{1ex}
        \item Random Pauli Rotation $\nu_{\mathsf{RPR}}$ on $\SU(2^n)$~\cite{haah2025efficient}.
        \vspace{1ex}
            
            Choose a Pauli operator $P \in \calP_n$ uniformly at random
            and choose an angle $\theta \in \RR/2\pi \ZZ$ uniformly at random.
            Apply $e^{\ii\theta P} \in \SU(2^n)$.
            This walk is equivalent to a single-qubit rotation $e^{\ii \theta \sigma^z_1}$
            with $\theta \sim \mu(\RR/2\pi\ZZ)$
            conjugated by a uniformly random Clifford.
            
        \vspace{1ex}
        \item Random Pauli Clifford Rotation $\nu_{\mathsf{RPCR}}$ on $\Cl(n)$.
        \vspace{1ex}
            
            Choose a Pauli operator $P \in \calP_n$,
            an angle $\theta \in \{ m \frac \pi 4 ~:~ m = 0,1,2,\ldots, 7\}$,
            and an element~$c$ from $\LCl(\emptyset,n)$,
            all uniformly at random.
            Apply $c e^{\ii\theta P} \in \Cl(n)$.
            This walk is equivalent to a single-qubit rotation $c e^{\ii \theta \sigma^z_1}$
            with $\theta$ a multiple of~$\pi/4$,
            conjugated by a uniformly random Clifford.

        \vspace{1ex}
        \item Brickwork Random Unitary Circuit $\nu_{\mathsf{BRUC}}$ on $\SU(2^n)$.
        \vspace{1ex}

            Choose $n-1$ unitaries $U_1, U_2, \ldots, U_{n-1}$ uniformly at random from
            two-qubit local unitary group~$\LU(2)$ independently.
            Apply $U_{2j-1}$ on two qubits~$2j-1, 2j$ and then 
            $U_{2j}$ on two qubits~$2j, 2j+1$ for all~$j$.
            This is a unitary circuit of depth~$2$.
            
        \vspace{1ex}
        \item Brickwork Random Clifford Circuit $\nu_{\mathsf{BRCC}}$ on $\Cl(n)$.
        \vspace{1ex}

            Choose $n-1$ Clifford unitaries $U_{i}$ uniformly at random
            from the two-qubit local Clifford group~$\LCl(2)$ independently.
            Apply $U_{2j-1}$ on two qubits~$2j-1, 2j$ and then 
            $U_{2j}$ on two qubits~$2j, 2j+1$ for all~$j$.
            This is a Clifford circuit of depth~$2$.
        
    \end{itemize}
\end{definition}

\begin{theorem}\label{thm:main}
    The spectral gaps are bounded as the following.
    \begin{enumerate}
        \item $\Delta(\nu_{\mathsf{RPR}}, \SU(2^n)) > 2^{-4}$ (see~\ref{thm:RPR}.)
        \item $\Delta(\nu_{\mathsf{RPCR}}, \Cl(n)) > 2^{-3}$ (see~\ref{thm:RPCR}.)
        \item $\Delta(\nu_{\mathsf{BRUC}}, \SU(2^n)) > 2^{-53}$ (see~\ref{thm:BRUC}.)
        \item $\Delta(\nu_{\mathsf{BRCC}}, \Cl(n)) > 2^{-7}$ (see~\ref{thm:BRCC}.)
    \end{enumerate}
\end{theorem}
This theorem implies that the Random Pauli Rotation and Brickwork Random Unitary Circuit
converge to approximate unitary $t$-designs after $\calO(nt)$ steps:
\begin{lemma}[{\cite{brandao2016local},\cite[Lemma 2.6]{chen2025incompressibility}}]
    Suppose that the spectral gap of $\nu$ on~$\SU(2^n)$
    in the balanced tensor representation
    $\tau_{t,t}: U \mapsto U^{\otimes t} \otimes \bar U^{\otimes t}$
    is $\Delta$.
    Let $A,B$ be any positive semi-definite operators
    acting on $(\CC^2)^{\otimes nt } \otimes \CC^d$,
    and write $W_t = W^{\otimes t} \otimes \id$ for any $W \in \SU(2^n)$.
    If $k \ge \frac{1}{\Delta} (nt \log 4 + \log \frac{1}{\epsilon})$,
    then 
    \begin{equation}
         (1-\epsilon) \avg_{V \sim \SU(2^n)} \Tr(A V_t B V_t^\dagger)
         \le
         \avg_{U \sim \nu^{*k}}\Tr(A U_t B U_t^\dagger)
         \le (1+\epsilon)  \avg_{V \sim \SU(2^n)} \Tr(A V_t B V_t^\dagger) .
    \end{equation}
\end{lemma}
We also prove analogous constant-gap results
for orthogonal groups, which imply convergence to orthogonal $t$-designs.
In this case, the implied mixing time with respect
to a constant additive and multiplicative error (say $0.01$)
is optimal (up to an absolute constant factor) for~$t=2$~\cite[\S2.4]{schuster2025random}.

Our results immediately improve implications for
the complexity growth of random unitary circuits as a function of the depth.
The argument reviewed in~\cite{chen2025incompressibility} applies almost verbatim,
but previously the depth had to be~$\mathcal{O}(2^{cn})$ for some~$c < 1$.
Our unconditional constant spectral gap removes this restriction.
\begin{corollary}[See \ref{cor:UnitaryComplexityGrowth}]
    Let $\delta \in (0,\frac 1 2)$ and $\gamma \in (0,1)$ be constants
    and $U$ be a random circuit consisting of $L \le 4^{(1-\gamma) n}$ independent random samples 
    from the \emph{All-to-All Random Gate} ensemble
    \begin{equation}
        \nu_{\mathsf{All}} = \frac{2}{n(n-1)} \sum_{i < j} \mu(\LU(\{i,j\},n))\, .
    \end{equation}
    The unitary circuit complexity satisfies
    \begin{equation}
        C_{\delta}(U) = \Omega \left(\frac{L}{n \log (nL)}\right)
    \end{equation}
    with probability at least $1 - \exp(-\Omega(L/n))$ over the choice of~$U$.
\end{corollary}

\subsection{Previous work}

To the best of our knowledge,
all previous spectral gap bounds~\cite{brown2010convergence,bourgain2012spectral,brandao2016local,haferkamp2021improved,
haferkamp2022random,metger2024simple,chen2024efficient,haah2025efficient,chen2025incompressibility}
for any walk with \emph{efficiently implementable} steps on the projective unitary 
group~$\PU(2^n)$
were vanishing in the limit of large representations or in the group dimension.
Here, the efficiency means that each step can be written 
as a product of~$\mathrm{poly}(n)$ two-qubit unitaries.
For balanced tensor representations $U \mapsto (U \otimes \bar U)^{\otimes t}$,
the slowest decaying previous bound was given in~\cite{chen2025incompressibility},
which reads~$1/\mathrm{polylog}(t)$ provided that~$t \lesssim 2^{cn}$ for some constant $c < 1$.
Our result is the first that gives a constant spectral gap,
independent of $n$ and representations,
and each step is manifestly efficient.
Some arguments~\cite{hunter2019unitary,Jian2023} appealing to statistical physics intuition
have indicated that the spectral gap of the Brickwork Random Unitary Circuit is a constant,
which we prove here.

Our results are inspired by the study of Kac's random walk.
A step of Kac's random walk~\cite{kac1956foundations} on~$\SO(m)$ 
uniformly chooses two indices then applies a rotation uniformly from $\SO(2)$.
\cite{maslen2003eigenvalues,carlen2003determination} proved that the spectral gap is realized by the $t=2$ case, 
which was noted by the best known bound for random unitary circuits \cite{chen2025incompressibility}.
A few years after these results,
Caputo~\cite{caputo2008spectral} provided an elementary proof using a local to global principle.
Recently, \cite{alon2026aldous} proved a similar spectral gap result 
for an analogous random walk on $\SU(2^n)$.
However, this gap is exponentially small in the number of qubits 
and the walk does not use the tensor product structure.

Caputo's elementary argument~\cite[Theorem~1.1]{caputo2008spectral} for Kac's random walk
is a variant of Knabe bounds~\cite{knabe1988energy}.
For one-dimensional translation-invariant frustration-free Hamiltonians,
Knabe bounds boost the spectral gap of a local subsystem to that of the entire system, 
as long as the local gap is above the so-called Knabe threshold.
Knabe~\cite{knabe1988energy} introduced this technique 
to generalize the seminal result that the AKLT chain is gapped~\cite{affleck1988valence} 
to other spin chains.
The technique enjoyed subsequent refinements~\cite{gosset2016local,lemm2019spectral}.

Successful application of Knabe's method~\cite{knabe1988energy} is delicate, requiring the local subsystem gap to be above threshold.
Haferkamp and Hunter-Jones~\cite{haferkamp2021improved} 
applied Knabe's method to unitary designs to prove 
that the Brickwork Random Unitary Circuit has a constant gap when the design order $t$ is small compared to the local dimension.
Without this assumption, they provide numerical evidence
that the local subsystem gap is exactly at threshold,
giving a vacuous result.

Our approach is different in that we apply the Knabe method to two random walks,
one on Clifford groups via Brickwork Random Clifford Circuit and 
the other on~$\SU(2^n)$ via Random Pauli Rotation~\cite{haah2025efficient},
and combine them using certain boosting methods~\cite{brandao2016local,chen2025incompressibility}.
The combination of random $n$-qubit Cliffords and random single-qubit unitaries
has appeared in~\cite{Zhou_2020,haferkamp2023efficient}.
Haferkamp~\cite{haferkamp2022random} applied 
the comparison method~\cite{DiaconisSaloffCoste1993} to the Clifford group
(which was previously also used by DiVincenzo, Leung, and Terhal~\cite{DiVincenzo_2002}),
and obtained a lower bound $\Omega(n^{-3})$ on the spectral gap of Brickwork Random Clifford Circuit
and used it in a random Clifford-plus-single-qubit-unitary walk via the coupling method~\cite{Oliveira_2009},
to derive a spectral gap lower bound~$\Omega(4^{-n} n^{-5})$ for Brickwork Random Unitary Circuit.
The estimate $\Omega(4^{-n} n^{-5})$ was an improvement over
the result of Brandao, Harrow, and Horodecki~\cite{brandao2016local},
who also used the coupling method~\cite{Oliveira_2009}.
These bounds are independent of representations of the respective groups,
but still depends on the ambient group order or dimension.

We prove constant lower bounds (and hence optimal up to constants)
for these walks using the Knabe method.
Recently, Knabe bounds have been generalized to Hamiltonians
in multiple dimensions \cite{lemm2019finite,lemm2022quantitatively, anshu2020improved}
and on arbitrary graphs \cite{mittal2023local,hunter2025two}.
For our results, the original Knabe bound \cite{knabe1988energy}
and the triangle counting of~\cite{caputo2008spectral} are sufficient, 
the latter of which implicitly identifies each local subsystem with a single qubit, 
which has a gap above the threshold for the Random Pauli Rotation.

For unitary $t$-designs,
it takes dimensional factors (that are exponential in~$n$ and $t$)
to convert the spectral gap of a random walk on~$\PU(2^n)$
to an operationally meaningful approximation error guarantees.
Indeed, our spectral gap lower bound for Brickwork Random Unitary Circuit in~\ref{thm:BRUC}
implies that Brickwork Random Unitary Circuit of depth $\mathcal O(nt + \log \frac 1 \epsilon)$
is an approximate unitary $t$-design with multiplicative error~$\epsilon$~\cite{brandao2016local}.
``Shallow unitary designs''~\cite{schuster2025random,laracuente2026approximate}
do not analyze spectral gaps
but give approximate unitary $t$-designs in multiplicative error~$\epsilon$
as some random unitary circuit of
depth~$t\, \mathrm{polylog}(nt/\epsilon)$.

\subsection{Subsequent work}

After the first version of this manuscript has appeared,
closely related results were announced.
In~\cite{liu2026almostquadruplyoptimalunitary},
an approximate unitary $t$-design on $n$ qubits with multiplicative error~$\epsilon$ is constructed
using a random circuit of depth~$\calO(\log\frac n \epsilon + t \log t)$ provided $t = \calO(n)$.
See also~\ref{rem:orthogonalMixingTime}.
In addition, they found a constant-sized constant-depth
generating set for the $n$-qubit Clifford group with a constant spectral gap.
In~\cite{anshu2026depth1expandersunitarygroup},
our result on Random Pauli Rotations is leveraged
to give a constant-sized generating set for $\PU(2^n)$ consisting of depth-1 circuits,
which has a constant spectral gap.
They also give a depth-1 constant-sized generating set for the Clifford group with a constant spectral gap.
In~\cite{li2026gatetimecomplexitygrowth},
the implication from our result to unitary complexity growth
is refined. The author does not use our result on spectral gap,
but gives more direct estimates for collision probabilities. 
The resulting complexity lower bound 
is a factor of~$n$ higher than~\ref{cor:UnitaryComplexityGrowth}.
In~\cite{chen2026spectralgapdownupwalks},
the authors streamline
the so-called trickle-down method for spectral gap analysis of Markov chains.
This work is completely independent of our work,
but their approach has some similarity with the Knabe method.
In \cite{dong2026spectral}, they claim to exactly calculate the spectral gap of the Random Pauli Rotation and improve that of the Brickwork Random Unitary Circuit.

\subsection{Proof ideas}
For the Brickwork Random Clifford Circuit and the Random Pauli Rotation, we exactly calculate the gap on a small subsystem and extend this to a constant gap on an arbitrarily large system.
The local to global principle we use is called Knabe's method~\cite{knabe1988energy} and a variant called triangle counting~\cite{caputo2008spectral}.
We combine these two results using certain boosting methods~\cite{brandao2016local,chen2025incompressibility} to show a constant gap for the Brickwork Random Unitary Circuit.

The simplest example of Knabe's method is a sum of projectors,
$\calH = \sum_{i=1}^{n-1} \Pi_{i}$, where $\Pi_i$ acts on qubits $i,i+1$.
We lower bound the gap of~$\calH$:
\begin{align}
    \nonumber\calH^2 
    &= \sum_{i=1}^{n-1} \Pi_i
        +\sum_{i=1}^{n-2} (\Pi_i \Pi_{i+1} + \Pi_{i+1} \Pi_i)
        + \sum_{|i-j|\ge 2} \Pi_i \Pi_j\\
    &\succeq 
    \calH + \sum_{i=1}^{n-2} (\Pi_i \Pi_{i+1} + \Pi_{i+1} \Pi_i) & \text{($\Pi_i \Pi_j = \Pi_j \Pi_i \succeq 0$ if $|i-j|\ge 2$)}
\end{align}
Consider the smallest nontrivial case of this problem, when $n=3$.
Suppose that this gap is $\Delta_\mathrm{local}$,
\begin{equation}\begin{split}
    (\Pi_i + \Pi_{i+1})^2  &\succeq \Delta_\mathrm{local} (\Pi_i + \Pi_{i+1}),\\
    \Pi_i \Pi_{i+1} + \Pi_{i+1}\Pi_i  &\succeq (\Delta_\mathrm{local}-1) (\Pi_i + \Pi_{i+1})\, .
\end{split}\end{equation}
Then,
\begin{align}
    \calH^2 
    &\succeq \calH + (\Delta_\mathrm{local} - 1)\sum_{i=1}^{n-2} (\Pi_i + \Pi_{i+1})
    = (2\Delta_\mathrm{local} - 1)\calH + (1-\Delta_\mathrm{local})(\Pi_1 + \Pi_{n-1})\\
    &\succeq (2\Delta_\mathrm{local} - 1) \calH\, . \nonumber
\end{align}

This example applies directly to a variant of the Brickwork Random Clifford Circuit, using the unnormalized moment operator as $\calH$.
Crucially, we find that the small subsystem gap on $\LCl(3)$ is above threshold.
To calculate this, diagonalization in exact rational arithmetic should be straightforward,
but $\LCl(3)$ is rather a large group that naive diagonalization seems difficult.
Instead, we prove generally that a normal subgroup can be factored out while preserving
almost all of the eigenspectrum, which reduces the problem to $\Sp(6;\FF_2)$.
We further observe that there is a Gelfand pair that is relevant for the calculation,
which reduces diagonalization to calculation
of a sum of character values over a subgroup of~$\Sp(6;\FF_2)$. 
See~\S\ref{sec:CliffordCircuits}.

For the Random Pauli Rotation, we use triangle counting~\cite{caputo2008spectral}.
The small subsystem is taken to be a $\su(2)$-triangle,
which is a subalgebra of~$\su(2^n)$ that is isomorphic to~$\su(2)$,
spanned by a triple of pairwise anticommuting Pauli operators.
This small subsystem gap was previously calculated exactly~\cite{haah2025efficient}
and is above threshold.
The triangle counting is especially simple
using the properties of the Pauli anticommutation graph.
See~\S\ref{sec:RandomPauliRotations}.

\begin{figure}
    \centering
    \newcommand{\figscale}{0.01}
    \resizebox{0.80\textwidth}{!}{\begin{tikzpicture}[
  x=1cm, y=1cm,
  blk/.style  = {draw, line width=0.8pt, inner sep=0pt, font=\large, fill=white},
  tq/.style   = {draw, line width=0.7pt, fill=white, inner sep=0pt,
                 minimum width=3.4mm, minimum height=5.6mm},
  wire/.style = {line width=0.4pt},
  lbl/.style  = {font=\small},
  brc/.style  = {decorate, decoration={brace, mirror, amplitude=4pt}, line width=0.6pt},
  sq/.style   = {-{Stealth[length=4.5pt]}, line width=0.7pt, decorate,
                 decoration={snake, amplitude=0.6mm, segment length=2.6mm,
                             pre length=1mm, post length=2.4mm}},
  lead/.style = {-{Stealth[length=3.5pt]}, line width=0.6pt},
]

\def\ytop{0.11}   \def\ybot{-1.81}   \def\ymid{-0.85}   \def\ybrace{-1.93}
\def\dx{0.44}     

\newcommand{\wires}[2]{\foreach \i in {1,...,6}{%
  \draw[wire] (#1,{-0.34*(\i-1)}) -- (#2,{-0.34*(\i-1)});}}
\newcommand{\brick}[2]{\node[tq] at (#1,{-0.34*(#2-1)-0.17}) {};}
\newcommand{\evenlayer}[1]{\brick{#1}{1}\brick{#1}{3}\brick{#1}{5}}
\newcommand{\oddlayer}[1]{\brick{#1}{2}\brick{#1}{4}}
\newcommand{\blockIV}[1]{%
  \evenlayer{#1}\oddlayer{#1+\dx}\evenlayer{#1+2*\dx}\oddlayer{#1+3*\dx}}
\newcommand{\blockVIII}[1]{\blockIV{#1}\blockIV{#1+4*\dx}}
\newcommand{\bigblock}[3]{\node[blk, anchor=north west, minimum width=#2 cm,
                          minimum height=1.92cm] at (#1,\ytop) {#3};}
\newcommand{\labbrace}[3]{\draw[brc] (#1,\ybrace) -- (#2,\ybrace)
                          node[midway,below=5pt,lbl]{#3};}

\newcommand{\sigmagate}[1]{%
  \draw[line width=0.7pt, fill=white] (#1-0.26,\ytop) rectangle (#1+0.26,-0.15);
  \node at (#1,-0.02) {\scalebox{0.8}{\scriptsize$\sigma_1^{z}$}};}

\newcommand{\panelnum}[2]{\node[font=\small\bfseries, anchor=north]
                           at (#1,-2.38) {(#2)};}

\newcommand{\callout}[2]{%
  \node[lbl, anchor=south, inner sep=1pt] at (#1,0.42) {#2};
  \draw[lead] (#1,0.38) -- (#1,0.17);}

\wires{0.00}{2.70}
\bigblock{0.00}{0.95}{$C$}
\sigmagate{1.35}
\bigblock{1.75}{0.95}{$C^{\dagger}$}
\labbrace{0.00}{0.95}{$\mathrm{Cl}(n)$}
\labbrace{1.75}{2.70}{$\mathrm{Cl}(n)$}

\panelnum{1.35}{1}

\draw[sq] (2.90,\ymid) -- (3.50,\ymid);

\wires{3.70}{6.40}
\bigblock{3.70}{0.95}{$C$}
\sigmagate{5.05}
\bigblock{5.45}{0.95}{$D$}
\labbrace{3.70}{4.65}{$\mathrm{Cl}(n)$}
\labbrace{5.45}{6.40}{$\mathrm{Cl}(n)$}

\panelnum{5.05}{2}

\draw[sq] (6.60,\ymid) -- (7.20,\ymid);

\wires{7.38}{11.49}
\blockIV{7.55}
\sigmagate{9.40}
\blockIV{10.00}
\labbrace{7.38}{9.04}{BRCC}
\labbrace{9.83}{11.49}{BRCC}
\callout{7.55}{$\mathrm{Cl}(2)$}
\panelnum{9.44}{3}

\draw[sq] (9.40,-2.80) -- (9.40,-3.58);

\begin{scope}[yshift=-3.95cm]
  \wires{7.38}{11.49}
  \blockIV{7.55}
  \sigmagate{9.40}
  \blockIV{10.00}
  \labbrace{7.38}{9.04}{BRUC}
  \labbrace{9.83}{11.49}{BRUC}
  \callout{7.55}{$\mathrm{SU}(4)$}
  \panelnum{9.44}{4}

  \draw[sq] (11.70,\ymid) -- (12.30,\ymid);

  \wires{12.48}{15.73}
  \blockVIII{12.65}
  \labbrace{12.48}{15.73}{BRUC}
  \callout{12.65}{$\mathrm{SU}(4)$}
  \panelnum{14.11}{5}
\end{scope}

\end{tikzpicture}}
    \caption{A pictorial view of the deformation of random unitary circuits from~\ref{thm:AdditiveRandomUnitary}.
    (1) The Random Pauli Rotation has a constant spectral gap.
    (2) The Clifford adjoint is substituted with an independent Clifford.
    (3) The two Cliffords are substituted with BRCCs, since those have a constant spectral gap.
    (4) The two-qubit Cliffords are enlarged to two-qubit unitaries.
    (5) The single-qubit rotation is absorbed into a two-qubit unitary.
    This image was hand-drawn and then uploaded to Claude Chat (Opus 5) for transcription to a TikZ script.
    }
    \label{fig:deformation}
\end{figure}
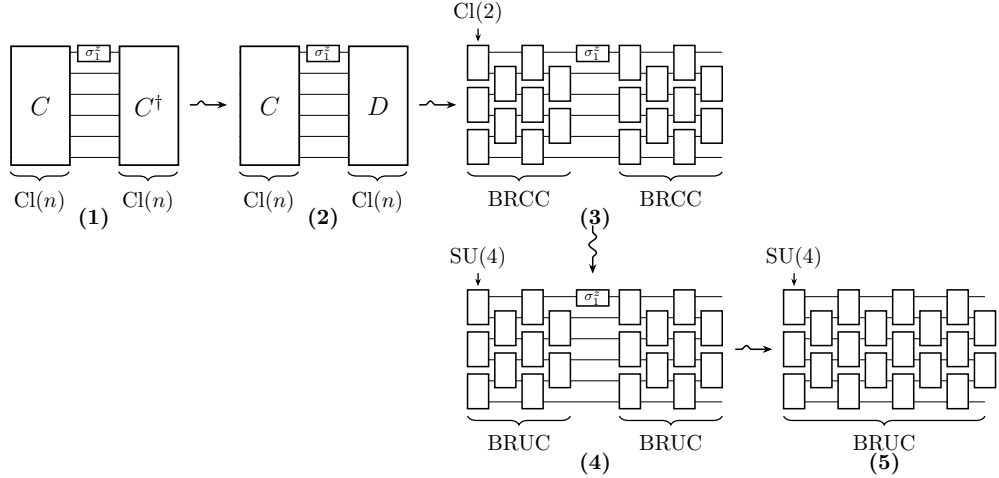

These two results combined with certain boosting methods~\cite{brandao2016local,chen2025incompressibility},
especially the detectability lemma and Gao's converse~\cite{aharonov2009detectability,
gao2015quantum,anshu2016simple,o2022quantum},
imply a constant gap for the Brickwork Random Unitary Circuit.
The flow of the proof is illustrated in Figure~\ref{fig:deformation}.
As the Random Pauli Rotation is a single-qubit rotation
conjugated by a uniformly random Clifford,
we replace the Cliffords by a Brickwork Random Clifford Circuit of some constant depth.
The spectral gap is reduced but only by a small constant.
Once the circuit consists of local gates,
we enlarge the two-qubit Cliffords and single-qubit rotation to two-qubit unitaries,
which does not reduce the spectral gap.
The remaining circuit is a Brickwork Random Unitary Circuit with constant depth,
which inherits a constant gap.
See~\S\ref{sec:BrickworkRandomUnitary}.

\subsection{Conventions and remarks}

All representations are finite dimensional unitary.
For a representation~$\rho$ of a group~$G$ 
with the Haar probability measure~$\mu(G)$,
we define
\begin{equation}
    \Pi(\rho|G) = M(\mu(G), \rho, G) = \int \rho(g) \,\rd \mu(g)
\end{equation}
which is an orthogonal projector onto the trivial subrepresentation of~$G$;
it may be zero.
We often write simply $x \sim G$ to mean $x \sim \mu(G)$
for a compact group~$G$.
All norms are the operator norm, the largest singular value of a matrix.

\begin{remark}
    In the context of unitary $t$-designs 
    (see e.g.~\cite{chen2025incompressibility} and references therein),
    only balanced tensor representations 
    $\tau_{t,t} : \SU(2^n) \ni U \mapsto U^{\otimes t} \otimes \bar U^{\otimes t}$ 
    are important.
    This representation factors through 
    the projective unitary group $\PU(2^n) = 
    \SU(2^n)/\{ \phi \one : \phi \in \CC, \abs{\phi} = 1\}$,
    but not all irreps of~$\SU(2^n)$ appear in balanced tensor representations.
    However, when it comes to spectral gap questions on random walks
    where each random step is based on subgroups containing the center of~$\SU(2^n)$,
    it suffices to consider balanced tensor representations;
    see~\ref{thm:projectiveFactorThrough}.
    In any case,
    the following characterizations of spectral gap are in fact equivalent.
    \begin{enumerate}
        \item The spectral gap of $\nu$, $\Delta(\nu,\SU(2^n))$.
        
        \item The spectral gap of $\nu$ in the regular representation, $\Delta(\nu,\rho_{\text{reg}},\SU(2^n))$.
        The regular representation of a possibly infinite compact group~$G$ is on~$L^2(G)$
        and is
        given by $(\rho_\text{reg}(V)f)(U) = f(V^{-1}U)$ for all $U,V \in G$ and $f \in L^2(G)$.
        The Peter--Weyl theorem states that the regular representation 
        decomposes into a direct sum of all finite dimensional irreps.
        
        \item The spectral gap of $\nu$ over all of the tensor representations, $\inf_{t+s>0}\, \Delta(\nu,\tau_{t,s},\SU(2^n))$. Every finite dimensional irrep of $\SU(2^n)$ appears in a tensor representation for some $t,s$~\cite{koike1989decomposition,alon2026aldous}.
        
        \item The spectral gap of the Hecke operator on $L^2(\SU(2^n))$ defined by $\nu$.
    \end{enumerate}
    In our arguments, tensor representations play no role except for calculation in~\ref{lem:upper}.
\end{remark}

\begin{remark}[{\cite[Remark $2.3$]{chen2025incompressibility}}]
\label{remark:RandomWalks}
    A random walk on a group $G$ induces a probability measure,
    \begin{equation}
        \nu = \nu_1 * \cdots * \nu_n,
    \end{equation}
    where step $i$ is drawn from a probability measure $\nu_i$.
    The corresponding moment operator follows a similar rule,
    \begin{equation}
        M(\nu_1 * \nu_2,\rho) = \Ex_{g\sim \nu_1, h \sim \nu_2} \rho(gh) = \left(\Ex_{g\sim \nu_1} \rho(g) \right) \left( \Ex_{h\sim \nu_2} \rho(h)\right) = M(\nu_1, \rho) M(\nu_2, \rho).
    \end{equation}
    The essential norm of this moment operator contracts exponentially with the number of steps,
    \begin{equation}\begin{split}
        \left(M(\nu,\rho) - M(\mu(G),\rho)\right)^k &= M(\nu^{*k}, \rho) - M(\mu(G),\rho),\\
        g(\nu,\rho)^k &\ge g(\nu^{*k}, \rho),
    \end{split}\end{equation}
    where the equality follows from the fact that $M(\nu, \rho) M(\mu(G),\rho) = M(\mu(G), \rho) = M(\mu(G), \rho) M(\nu, \rho)$ by Haar measure left and right invariance.
\end{remark}

\begin{remark}
    Let $H_1, \ldots, H_{L}$ be subgroups of unitaries, each supported on some pair of qubits.
    Consider two ways to combine their Haar probability measures, by convolution and by average:
    \begin{equation}
        \convo = \mu(H_1) * \cdots * \mu(H_{L}) \quad \text{ and } 
        \quad \Sigma = \frac {\mu(H_1) + \cdots + \mu(H_{L})} {L} \, .
    \end{equation}
    The operational meaning is that, in $\convo$, 
    we choose $L$ elements $U_j \in H_j$ uniformly at random independently
    and apply them in sequence,
    while, in $\Sigma$, we choose an index~$j$ among $L$ choices uniformly at random
    and then choose an element $U_j \in H_j$ uniformly at random.
    Of course, it is not special here that $H_j$ is a unitary group of two-qubit unitaries.
    \begin{itemize}
    
        \item Brickwork Random Unitary Circuit of depth~$2$ on $\SU(2^n)$ corresponds to a convolution
        with subgroups $H_j = \LU(\{j,j+1\},n)$ 
        that are the full two-qubit unitary group acting on qubits $j,j+1$,
            \begin{equation}
                \nu_{\mathsf{BRUC}} = \underbrace{\left(\mu(H_1) * \mu(H_3) * \cdots \right)}_{\text{layer 1}} * \underbrace{\left(\mu(H_2) * \mu(H_4) * \cdots \right)}_{\text{layer 2}}.
            \end{equation}
            Using the same collection of subgroups~$H_j$,
            we also consider the average version:
            $\avg_j \mu(H_j)$.

        \vspace{1ex}
        \item All-to-All Random Gate corresponds to
        the average 
        \begin{equation}
            \nu_{\mathsf{All}} = \frac{2}{n(n-1)} \sum_{i < j} \mu(\LU(\{i,j\},n))\, .
        \end{equation}        
        We consider this in~\ref{thm:additiveClifford} and \ref{thm:AdditiveRandomUnitary},
        but they are simple applications of other results.
    
        \vspace{1ex}
        \item Random Pauli Rotation on $\SU(2^n)$ 
        is the average of Haar probability measures on $4^n -1$
        groups~$H_P = \{ e^{\ii \theta P} : \theta \in \RR/2\pi \ZZ\}$,
        each of which is isomorphic to~$U(1)$,
        \begin{equation}
            \nu_{\mathsf{RPR}} = \frac 1 {4^n -1} \sum_{P \in \calP_n} \mu(H_P) \, .
        \end{equation}
        It is closely related to but different from
        \begin{equation}
            \mu(\Cl(n)) * \mu(H_{\sigma^z_1}) * \mu(\Cl(n)) \, . \label{eq:CZC}
        \end{equation}
        In the latter, the two Clifford elements are independent,
        while in the Random Pauli Rotation they need to be inverse of each other.
        We use the ensemble~\eqref{eq:CZC} in the proof of~\ref{thm:AdditiveRandomUnitary}.
    \end{itemize}

    The Clifford versions of the above circuits can be categorized similarly.
\end{remark}

Given a collection of subgroups,
the averaged ensemble and convolution ensemble
are interchangeable by the detectability lemma~\cite{aharonov2009detectability} 
and its converse, Gao's quantum union bound~\cite{gao2015quantum},
whose proofs are simplified in~\cite{anshu2016simple,o2022quantum}.

\begin{lemma}[{\cite{anshu2016simple,o2022quantum}; transcribed in \cite[Lemma~2.20]{chen2025incompressibility}}]
    \label{lem:detectability}
    Let $H_1, \ldots, H_L$ be compact subgroups of a compact group~$G$,
    where each $H_i$ commutes element-wise with all but $\ell$ others.
    Suppose~$\ell \ge 1$.
    Consider the convolution and average of the subgroup Haar probability measures:
    \begin{equation}
        \convo = \mu(H_1) * \cdots * \mu(H_L) \quad \text{ and } 
        \quad \Sigma = \frac {\mu(H_1) + \cdots + \mu(H_L)} L \, .
    \end{equation}
    Then, for any representation~$\rho$ of~$G$,
    the spectral gaps in~$\rho$ (see \ref{defn:spectralgap} for the definition)
    are related as
    \begin{equation}
        \Delta(\Sigma, \rho) \ge \frac 1 {4L} \Delta(\convo, \rho) 
        \quad \text{ and }\quad 
        \Delta(\convo, \rho) \ge \frac 1 4 \min(1, L \ell^{-2}\Delta(\Sigma, \rho)) \, .
    \end{equation}
\end{lemma}

A nonzero spectral gap is synonymous to group generation:

\begin{lemma}[Spectral gap and generation]\label{lem:SpectralGapAndGeneration}
    Let $G$ be a finite or compact Lie group,
    and $A,B$ closed subgroups of~$G$.
    Then, $A$ and $B$ densely generate $G$
    if and only if
    \begin{equation}
        \norm*{ \frac{\Pi(\rho|A) + \Pi(\rho|B)}{2} - \Pi(\rho|G) } < 1
    \end{equation}
    for every representation~$\rho$ of~$G$.
\end{lemma}

\begin{proof}
    Suppose that $A$ and $B$ densely generate~$G$.
    Suppose on the contrary to the claim 
    that there is an irrep~$\rho$ in which the norm is~$1$.
    Then, this irrep must be nontrivial so $\Pi(\rho|G) = 0$,
    since on the trivial irrep the norm is automatically zero.
    The norm~$\norm{(\Pi(\rho|A)+\Pi(\rho|B))/2} = \frac 1 2 \sup_{\psi} \bra\psi (\Pi(\rho|A)+\Pi(\rho|B)) \ket \psi = 1$ is attained by a normalized vector~$\ket \psi$ in the representation space,
    which must be a common eigenvector of eigenvalue~$+1$.
    That is, $\ket \psi$ is $A$-invariant and $B$-invariant.
    By the dense generation and continuity of~$\rho$,
    for every element~$g \in G$ and $\epsilon >0$,
    there is a finite product $a_1 b_1 \cdots a_k b_k$ with $a_j \in A$ and $b_j \in B$
    such that $\norm{\rho(a_1)\rho(b_1) \cdots \rho(a_k)\rho(b_k) - \rho(g)} < \epsilon$.
    Acting on~$\ket \psi$, we see that $\norm*{\ket \psi - \rho(g) \ket \psi}_{2} < \epsilon$.
    Since $\epsilon$ was arbitrary, $\ket \psi$ must also be $G$-invariant,
    which is a contradiction to the assumption that~$\rho$ is a nontrivial irrep.

    Conversely, suppose that $A, B$ generate a proper closed subgroup~$H$ of~$G$.
    Then, $G$ acts nontrivially on~$L^2(G/H)$.
    which contains a trivial $G$-representation
    with multiplicity~$1$ since $G$ acts transitively.
    Since $H$ is a proper subgroup,
    the orthogonal complement~$W$ of the trivial $G$-representation
    is nonzero and consists of only nontrivial $G$-representations.
    The representation~$L^2(G/H)$ is the induced representation of~$G$
    from the trivial representation of~$H$,
    and 
    the Frobenius reciprocity~\cite[Example~6.3]{BrockertomDieck}
    (a finite group case handled in~\cite[Corollary~3.20]{fulton2013representation})
    implies that any $G$-irrep $V$ (which is necessarily finite dimensional by the Peter--Weyl theorem)
    in~$W$ must contain a trivial $H$-representation.
    In~$V$ we have $\Pi(\rho|G) = 0$ for $V$ being a nontrivial $G$-irrep,
    but $\Pi(\rho|A),\Pi(\rho|B) \succeq \Pi(\rho|H)$ is nonzero.
\end{proof}

The following tool becomes useful in calculation.

\begin{lemma}[Factoring out normal subgroups]
    \label{lem:factoringNormal}
    Let $G$ be a finite or compact Lie group,
    and $A$ and $B$ be closed subgroups of~$G$, densely generating~$G$.
    Let $N \triangleleft G$ be a normal subgroup
    such that
    \begin{itemize}
        \item $N = N_1 N_2 N_3$ for some subgroups $N_1 \le A$, $N_2 \le A \cap B$, and $N_3 \le B$
        \item with element-wise commutativity:
            $[N_1,N_2]=[N_2,N_3]=[N_3,N_1]=[N_1,B]=[A,N_3] = \one$.
    \end{itemize}
    Let $\rho$ denote any representation of~$G$,
    and $\sigma$ denote any representation of~$G/N$.
    Then, either
    \begin{equation}
        \sup_\rho ~\norm*{
            \frac{\Pi(\rho|A) + \Pi(\rho|B)}{2} 
            - \Pi(\rho|G)
        } \le \frac 1 2
    \end{equation}
    or
    \begin{equation}
        \begin{split}
        &\sup_\rho ~\norm*{
            \frac{\Pi(\rho|A) + \Pi(\rho|B)}{2} 
            - \Pi(\rho|G)
        }\\
        &=
        \sup_\sigma ~\norm*{
            \frac {\Pi(\sigma\,|\, (A N)/N) + \Pi( \sigma\,|\,(B N)/N) }{2}
            - \Pi(\sigma\,|\,G/N)
        }
        \end{split}
    \end{equation}
\end{lemma}

\begin{proof}
    Suppose $\Pi(\rho|A)$ and $\Pi(\rho|B)$ commute.
    The eigenspectrum of $\frac 1 2 (\Pi(\rho|A) + \Pi(\rho|B))$ is
    contained in $\{0,\frac 1 2, 1\}$.
    The eigenspace of the eigenvalue~$1$, if any, 
    forms the trivial representation of~$\langle A, B \rangle$,
    which is assumed to be~$G$.
    Hence, the norm is at most~$1/2$.
    Therefore, we may assume that $\Pi(\rho|A)$ does not commute with $\Pi(\rho|B)$ for some irrep~$\rho$ of~$G$.
    
    Then, there must be a nontrivial ($2\times 2$) Jordan block 
    (a pair of possibly noncommuting projectors
    can be simultaneously block diagonalized
    with block size either $1\times 1$ or $2 \times 2$),
    on which $\Pi(\rho|A)$ and $\Pi(\rho|B)$
    do not commute and $\Pi(\rho|G)$ is zero.
    Then, $(\Pi(\rho|A) + \Pi(\rho|B))/2$
    is unitary conjugate to $\diag(\cos^2 \theta, \sin^2 \theta)$ for some~$\theta \in \RR$
    on the Jordan block;
    since $\Pi(\rho|A)$ does not commute with~$\Pi(\rho|B)$, 
    both components $\cos^2 \theta$ and $\sin^2 \theta$ must be nonzero.
    In particular, we have the essential norm
    \begin{equation}
        \norm*{\frac{\Pi(\rho|A) + \Pi(\rho|B)}{2} - \Pi(\rho|G)}
        =\bra \psi \left( \frac{\Pi(\rho|A) + \Pi(\rho|B)}{2} - \Pi(\rho|G) \right) \ket \psi > \frac 1 2
    \end{equation}
    for some normalized nonzero vector~$\ket \psi$ such that $\Pi(\rho|G) \ket \psi = 0$.
    
    Each projector $P_j = \Pi(\rho|N_j)$ for any~$j=1,2,3$ 
    commutes with both~$\Pi(\rho|A)$ and $\Pi(\rho|B)$.
    This is either because $N_j$ is a subgroup or because of the commutativity assumption.
    Therefore, we may assume that $\ket \psi$ is a common eigenvector of $P_j$ for all~$j=1,2,3$.
    If $P_1 \ket \psi = 0$, then we use $P_1 \Pi(\rho|A) = \Pi(\rho|A)$ to see that
    $\bra \psi ((\Pi(\rho|A) + \Pi(\rho|B))/2 - \Pi(\rho|G)) \ket \psi \le \frac 1 2$;
    if $P_2 \ket \psi = 0$, then we use $P_2 \Pi(\rho|A) = \Pi(\rho|A)$ to see that the essential norm is~$\le \frac 1 2$;
    if $P_3 \ket \psi = 0$, then we use $P_3 \Pi(\rho|B) = \Pi(\rho|B)$ to have the same contradiction.
    
    Therefore, we must have $P_1 \ket \psi = P_2 \ket \psi = P_3 \ket \psi = \ket \psi$.
    Since~$N$ is normal,
    this means that the common $+1$-eigenspace of~$P_j$ is a subrepresentation%
    \footnote{This step fails if $N$ is not normal.}
    of~$\rho$.
    Since $\rho$ is irreducible,
    this subrepresentation is the full~$\rho$,
    in which all~$N_j$ are trivially represented.
    This implies that $\ker \rho \supseteq N$,
    and $\rho$ defines a representation of~$G/N$.
    The lemma is proved.
\end{proof}

\begin{corollary}\label{thm:projectiveFactorThrough}
    The spectral gap of the Brickwork Random Unitary Circuit
    is equal to that of the induced random walk on the projective unitary group.
\end{corollary}

\begin{proof}
    Each walk is based on two closed subgroups
    that contains the center of the parent group.
    The claim follows from~\ref{lem:factoringNormal}
    by letting $N_1 = N_2 = N_3$ be the center.
\end{proof}

\section{Clifford circuits}
\label{sec:CliffordCircuits}

Define the symplectic group:
\begin{equation}
    \Sp(2n;\FF_2) = 
    \left\{
        \begin{bmatrix}
            A & B \\ C & D
        \end{bmatrix} \in \FF_2^{2n \times 2n} 
        ~\middle|~
        \begin{bmatrix}
            A^T & C^T \\ B^T & D^T
        \end{bmatrix}
        \begin{bmatrix}
            0_{n \times n} & \one_{n} \\ \one_n & 0_{n \times n}
        \end{bmatrix}
        \begin{bmatrix}
            A & B \\ C & D
        \end{bmatrix}
        = 
        \begin{bmatrix}
            0_{n \times n} & \one_{n} \\ \one_n & 0_{n \times n}
        \end{bmatrix}
    \right\}.
\end{equation}

\begin{proposition}\label{prop:sp6}
    Let $H_{1,2} \cong H_{2,3} \cong \Sp(4;\FF_2)$ be the images in~$\Sp(6;\FF_2)$ by
    \begin{equation}
        \begin{bmatrix}
            A & 0 & B & 0 \\
            0 & 1 & 0 & 0 \\
            C & 0 & D & 0 \\
            0 & 0 & 0 & 1
        \end{bmatrix}
        \xleftarrow{~~\text{for }H_{1,2}~~}
        \underbrace{\begin{bmatrix}
            A_{2\times 2} & B_{2\times 2} \\ C_{2\times 2} & D_{2\times 2} 
        \end{bmatrix}}_{\in\Sp(4;\FF_2)}
        \xrightarrow{~~\text{for }H_{2,3}~~}
        \begin{bmatrix}
            1 & 0 & 0 & 0 \\
            0 & A & 0 & B \\
            0 & 0 & 1 & 0 \\
            0 & C & 0 & D
        \end{bmatrix}
        \, .
    \end{equation}
    Then, for any representation~$\rho$ of~$\Sp(6;\FF_2)$
    we have
    \begin{equation}
        \norm*{
            \frac {\Pi(\rho|H_{1,2}) + \Pi(\rho|H_{2,3})}{2} - \Pi(\rho|\Sp(6;\FF_2))
        }
        \le
        1 - \frac{3}{10}
    \end{equation}
    where the equality holds for the permutation representation acting on~$\FF_2^6$.
    In fact, the eigenspectrum of 
    $\Pi(\rho_{reg}|H_{1,2}) + \Pi(\rho_{reg}|H_{2,3})$
    for the regular representation~$\rho_{reg}$ 
    is $\{0,\frac 6 {10}, \frac 8 {10}, \frac 9 {10}, 1, \frac{11}{10}, \frac{12}{10}, \frac{14}{10}, 2\}$.
\end{proposition}

The proof of~\ref{prop:sp6} enlarges the subgroup from $\Sp(4; \FF_2)$ to $\Sp(4; \FF_2) \times \Sp(2; \FF_2)$, while retaining (almost all) of the spectrum by \ref{lem:Sp4toSp4Sp2}.
The latter is known to form a Gelfand pair with $\Sp(6;\FF_2)$.
This property enables us to calculate the spectrum exactly using elementary character theory.
We explain how this calculation is performed by a computer with exact rational arithmetic
at the end of this section.

\begin{corollary}\label{cor:3qubitClifford}
    Let $n \ge 3$.
    Consider 3 qubits $a,b,c$.
    For any representation~$\rho$ of~$\LCl(\{a,b,c\},n)$,
    we have
    \begin{equation}
        \norm*{
            \frac{
                \Pi(\rho|\LCl(\{a,b\},n)) + \Pi(\rho|\LCl(\{b,c\},n))
            }{2}
            -
            \Pi(\rho|\LCl(\{a,b,c\},n))
        }
        \le 1 - \frac 3 {10} \, .
    \end{equation}  
\end{corollary}

\begin{proof}
    The group of all Pauli operators on the 3 qubits 
    is, by definition of Clifford unitary,
    a normal subgroup of~$\LCl(3)$.
    Moreover, it is a product of three Pauli groups, one on each qubit,
    that are element-wise commuting.
    By~\ref{lem:factoringNormal},
    the essential norm can be estimated by passing to the quotient group
    of the Clifford group by the Pauli group (including central phases).
    This quotient group is well known to be isomorphic 
    to the symplectic group~$\Sp(6,\FF_2)$~\cite[\S II.\,D]{Haah_2017}.
\end{proof}

\begin{theorem}\label{thm:additiveClifford}
    Consider $n$ qubits labeled by~$1,2,\ldots,n$.
    Let $\Cl_{a,b} = \LCl(\{a,b\},n)$ be local Clifford groups.
    Then, for every representation~$\rho$ of the full Clifford group~$\Cl(n)$,
    we have
    \begin{align}
        \label{eq:1d-random-Clifford}\norm*{
            \frac{\Pi(\rho|\Cl_{1,2}) + \Pi(\rho|\Cl_{2,3}) + \cdots + \Pi(\rho|\Cl_{n-1,n})}{n-1}
            -
            \Pi(\rho|\Cl(n))
         }
        &\le 1 - \frac{1}{5(n-1)} \, ,\\
        \label{eq:all-to-all-random-Clifford}\norm*{
            \frac{2}{n(n-1)}\sum_{i < j}\Pi(\rho|\Cl_{i,j}) 
            -
            \Pi(\rho|\Cl(n))
         }
        &\le 1 - \frac{1}{5(n-1)} \, . 
    \end{align}
\end{theorem}

\begin{proof}
    The second inequality~\ref{eq:all-to-all-random-Clifford}
    is a simple consequence of~\ref{eq:1d-random-Clifford}.
    Consider qubit permutation unitaries~$U_\sigma \in \Cl(n)$
    such that $U_\sigma \Cl_{i,j} U_\sigma^{-1} = \Cl_{\sigma(i),\sigma(j)}$,
    where $\sigma$ ranges over the symmetric group~$\symgp_n$.
    Since $\rho$ is a representation, $\rho(U_\sigma) \Pi(\rho|\Cl_{i,j}) \rho(U_\sigma^{-1}) = \Pi(\rho|\Cl_{\sigma(i),\sigma(j)})$.
    By Jensen's inequality,
    \begin{equation}\begin{split}
        \norm*{\avg_{\sigma \in \symgp_n} \Pi(\rho|\Cl_{\sigma(1),\sigma(2)})}
        &= \norm*{\avg_{\sigma \in \symgp_n} \avg_{j=1}^{n-1} \Pi(\rho|\Cl_{\sigma(j),\sigma(j+1)})}\\
        &\le \avg_{\sigma \in \symgp_n} \norm*{\avg_{j=1}^{n-1} \Pi(\rho|\Cl_{\sigma(j),\sigma(j+1)})}\\
        &= \avg_{\sigma \in \symgp_n} \norm*{ \rho(U_\sigma) \left(\avg_{j=1}^{n-1} \Pi(\rho|\Cl_{j,j+1})\right) \rho(U_\sigma^{-1})}\\
        &= \norm*{\avg_{j=1}^{n-1} \Pi(\rho|\Cl_{j,j+1})}.
    \end{split}
    \label{eq:JensenForAllToAll}
    \end{equation}
    This proves the implication from~\ref{eq:1d-random-Clifford} to~\ref{eq:all-to-all-random-Clifford}
    for $\rho$ not containing any trivial representation,
    and hence for any $\rho$ that may contain a trivial representation.

    To prove~\ref{eq:1d-random-Clifford},
    we write $\Pi_j = \one - \Pi(\rho|\Cl_{j,j+1})$ and $\calH = \sum_{j=1}^{n-1} \Pi_j$.
    We use Knabe's method~\cite{knabe1988energy} to lower bound the spectral gap of~$\calH$.
    \begin{align}
        \nonumber\calH^2 
        &= \sum_{i=1}^{n-1} \Pi_i + \sum_{|i-j|=1} \Pi_i \Pi_j + \sum_{|i-j|\ge 2} \Pi_i \Pi_j\\
        &\succeq \calH + \sum_{|i-j|=1} \Pi_i \Pi_j & \text{($\Pi_i \Pi_j = \Pi_j \Pi_i = (\Pi_i \Pi_j)^2 \succeq 0$ if $|i-j|\ge 2$)}\\
        \nonumber &= \calH + \sum_{i=1}^{n-2} (\Pi_i \Pi_{i+1} + \Pi_{i+1} \Pi_i).
    \end{align}
    By~\ref{cor:3qubitClifford} 
    with a representation~$\rho|_{\Cl(3)_{i,i+1,i+2}}$,
    we see that
    \begin{equation}\begin{split}
        (\Pi_i + \Pi_{i+1})^2  &\succeq \Delta (\Pi_i + \Pi_{i+1}) \qquad\qquad (\Delta = 3/5)\, ,\\
        \Pi_i \Pi_{i+1} + \Pi_{i+1}\Pi_i  &\succeq (\Delta-1) (\Pi_i + \Pi_{i+1})\, .
    \end{split}\end{equation}
    Then,\footnote{
        The Knabe method for this step
        naively applied for~\ref{eq:all-to-all-random-Clifford}
        would require $(n-1)(\Delta -1) + 1 > 0$ or $\Delta > (n-2)/(n-1)$.
        This is not fulfilled for $n > 3$.
    }
    \begin{align}
        \calH^2 
        &\succeq \calH + (\Delta - 1)\sum_{i=1}^{n-2} (\Pi_i + \Pi_{i+1})
        = (2\Delta - 1)\calH + (1-\Delta)(\Pi_1 + \Pi_{n-1})\\
        &\succeq (2\Delta - 1) \calH\, . \nonumber
    \end{align}
    One minus the norm in the claim is 
    the spectral gap of~$\calH$ divided by~$n - 1$.
\end{proof}

\begin{theorem}[implying ~\ref{thm:main}(4)]\label{thm:BRCC}
    The spectral gap of the Brickwork Random Clifford Circuit of depth~$2$
    on (the one-dimensional chain of) $n$ qubits for any $n \ge 3$
    is 
    \begin{equation}
        \Delta(\nu_{\mathsf{BRCC}}, \Cl(n)) \ge \frac 1 {80} \, .
    \end{equation}
\end{theorem}
\begin{proof}
    Let $\Sigma$ be a distribution on~$\Cl(n)$
    obtained by, first, choosing one out of $n-1$ ``bonds'' 
    (a bond is a pair of qubits~$i,i+1$) uniformly at random,
    and, second, choosing an element from~$\Cl_{i,i+1}$ uniformly at random.
    By~\ref{thm:additiveClifford}, we have~$\Delta(\Sigma) \ge 1/(5(n-1))$.
    There are $L = n - 1$ subgroups in the setting of~\ref{lem:detectability},
    and every subgroup commutes element-wise with all but $\ell =2$ other subgroups.
    Therefore, $\Delta(\convo) \ge \frac 1 4 \frac{n -1}{2^2} \frac {1}{5(n -1)} = \frac 1 {80}$.
\end{proof}

\subsection{Gap calculation}

    Since the regular representation of a finite group~$G$ contains all irreps,
    it suffices to find the top eigenvalue of a $\abs G \times \abs G$ self-adjoint matrix.
    Here, $G = \Sp(6;\FF_2)$ has order 
    $\abs G = (2^6 - 2^0)(2^5-2^1)(2^4-2^2)2^3 2^2 2^1 = 1\,451\,520$, a rather large number.
    It is impractical to generate $\abs G \times \abs G$ matrices with more than $2 \times 10^{12}$ entries
    and compute eigenvalues using rational arithmetic.
    The smaller symplectic group~$H = \Sp(4;\FF_2)$ has order~$720$,
    so the Lanczos method to extract top eigenvalues using approximate floating numbers 
    is numerically feasible.
    This gives a relatively quick and convincing result,
    but does not constitute a proof
    since the result is only approximate with weak accuracy guarantee.
    Fortunately, ATLAS~\cite{atlas} database contains valuable information about~$G$,
    and a character table is available in GAP system~\cite{GAP4}. 
    We provide self-contained and explicit calculation for our result,
    which is facilitated by, but does \emph{not} rely on, existing databases~\cite{atlas,GAP4}.
    We directly verify the character table,
    on which our calculation for~\ref{prop:sp6} is entirely based
    without any concrete representation of~$\Sp(6;\FF_2)$
    
    We find that it makes the calculation much simpler
    to consider a slightly larger generating set (Gelfand pairs) for~$\Sp(6;\FF_2)$.
    Let $F_{1,2} \cong F_{2,3} \cong \Sp(4;\FF_2) \times \Sp(2; \FF_2)$
    be the images in~$\Sp(6;\FF_2)$ by
    \begin{equation}
        \begin{bmatrix}
            A & 0 & B & 0 \\
            0 & a & 0 & b \\
            C & 0 & D & 0 \\
            0 & c & 0 & d
        \end{bmatrix}
        \xleftarrow{~~\text{for }F_{1,2}~~}
        \left(
        \underbrace{\begin{bmatrix}
            A_{2\times 2} & B_{2\times 2} \\ C_{2\times 2} & D_{2\times 2} 
        \end{bmatrix}}_{\in\Sp(4;\FF_2)} \oplus
        \underbrace{
        \begin{bmatrix}
            a_{1\times 1} & b_{1\times 1}\\ c_{1\times 1} & d_{1\times 1}
        \end{bmatrix}}_{\in\Sp(2;\FF_2)}
        \right)
        \xrightarrow{~~\text{for }F_{2,3}~~}
        \begin{bmatrix}
            a & 0 & b & 0 \\
            0 & A & 0 & B \\
            c & 0 & d & 0 \\
            0 & C & 0 & D
        \end{bmatrix}
        \, .
    \label{eq:Fembedding}
    \end{equation}
    
    \begin{lemma}
    \label{lem:Sp4toSp4Sp2}
        For any representation~$\rho$ of~$\Sp(6;\FF_2)$,
        \begin{equation}
        \norm*{\Pi(\rho|H_{1,2}) + \Pi(\rho|H_{2,3})} = \norm*{\Pi(\rho|F_{1,2}) + \Pi(\rho|F_{2,3})}
        \end{equation}
        if the left-hand side is different from~$1$.
    \end{lemma}
    \begin{proof}
        Suppose that $\Pi_1 \neq 0$ and $\Pi_2 \neq 0$ are two arbitrary projectors.
        By Jordan's lemma for two projectors,
        we can relate $\norm{\Pi_1 + \Pi_2}$ to $\norm{\Pi_1 \Pi_2}$
        for any two projectors~$\Pi_1$ and~$\Pi_2$:
        if $\theta \in [0,\frac \pi 2]$ is a principal angle
        such that $\norm{\Pi_1 \Pi_2} = \cos \theta$,
        then $\norm{\Pi_1 + \Pi_2} = 2 \cos^2 \frac \theta 2$.

        Observe that $\Pi(\rho | F_{1,2}) = \Pi(\rho|H_{1,2}) \Pi(\rho|F'_3)$
        where $F'_3$ is the embedding of~$\Sp(2;\FF_2)$ given
        by~\eqref{eq:Fembedding} with $\Sp(4;\FF_2)$ part set to identity.
        In addition, $\Pi(\rho|F'_3) \Pi(\rho | H_{2,3} ) = \Pi(\rho|H_{2,3})$
        since~$F'_3$ is a subgroup of~$H_{2,3}$.
        Hence, 
        \begin{equation}
            \Pi(\rho|F_{1,2})\Pi(\rho|F_{2,3}) = \Pi(\rho|H_{1,2}) \Pi(\rho|H_{2,3}).
            \label{eq:FFHH}
        \end{equation}
        (This equality is true even if~$\Pi(\rho|F_{1,2}) = 0$.)
        If the left-hand side is different from~$1$,
        then the principal angle must be different from~$\pi/2$,
        and the two projectors are both nonzero.
        Turning~\eqref{eq:FFHH} into the norm of the sum, we complete the proof.
    \end{proof}
  
\begin{proof}[Proof of~~\ref{prop:sp6}]
    Here we explain the calculation procedure
    with implementation in programming language Julia,
    which runs on a laptop for a few minutes.
    The step numbering is consistent with the Julia code files.

    \begin{enumerate}
        
        \item[Step i.] \emph{Explicit enumeration of~$G$}.
        
        The complete list of all elements of~$G$ is produced 
        by enumerating each line of a $6 \times 6$ symplectic matrix.
        Each element of~$G$ is encoded in a $36$-bit vector, stored as a $64$-bit integer.
        
        \vspace{1ex}
        \item[Step ii.] \emph{Partition of~$G$ into conjugacy classes}.
        
        Bring a copy of the full list of elements,
        and call it unclassified.
        Select any element,
        enumerate all the conjugate elements by going over $\abs G$ conjugations,
        record this conjugacy class~$C_\alpha$,
        and remove these classified elements from the unclassified list.
        Repeat this until the unclassified list becomes empty.
        Confirm that there are $r=30$ conjugacy classes.
        The enumeration of each conjugacy class takes $O(\abs G)$ time.
        Each conjugacy class is sorted (what ordering is unimportant),
        so it is fast to find which conjugacy class an element belongs to.

        \vspace{1ex}
        \item[Step iii.] \emph{Class algebra structure constants}.

        Calculate $c_{\alpha\beta\gamma} \in \ZZ_{\ge 0}$ defined by
        \begin{equation}
            C_\alpha \cdot C_\beta = \sum_{x \in C_\alpha, ~ y \in C_\beta} x y 
            = \sum_{\gamma} \sum_{z \in C_\gamma} c_{\alpha\beta\gamma} z = \sum_{\gamma} c_{\alpha\beta\gamma} C_\gamma \in \CC G \label{eq:classAlgebraConstants}
        \end{equation}
        where $\alpha,\beta,\gamma$ range over all conjugacy classes
        and we have abused notation that $C_\alpha$ denotes the class sum~$\sum_{g \in C_\alpha} g \in \CC G$
        in the group algebra.
        
        The algorithm to calculate $c_{\alpha\beta\gamma}$ is as follows.
        Initialize variables $c_{\alpha\beta\gamma}$ to zero;
        there are $r^3 = 27\,000$ variables.
        For each $C_\gamma$, pick a representative $x_\gamma$.
        For each $g \in G$, determine the conjugacy class $C_\alpha \ni g$ and $C_\beta \ni g^{-1} x_\gamma$
        and increase $c_{\alpha\beta\gamma}$ by~$1$.

        \vspace{1ex}
        \item[Step iv.] \emph{Character table verification}.

        We hard-code the character table of~$G$, which is downloaded from GAP system~\cite{GAP4},
        and verify that the table contains $r$ distinct characters of~$\Sp(6;\FF_2)$
        by
        \begin{proposition}[{\cite[Chapter~3]{Isaacs1976}}]\label{prop:central-character}
            For any unital algebra homomorphism~$\omega$ from the center of the group algebra,
            $\calZ(\CC G) = \CC\{C_\alpha\}$, of a finite group~$G$ to~$\CC$,
            there exists an irreducible character~$\chi$ of~$G$ such that
            \begin{equation}
                \omega(C_\alpha) = \frac{\chi(C_\alpha)}{\chi(\one)} \quad \text{ for all }\alpha\, .
            \end{equation}
            Conversely, every irreducible character gives a homomorphism from~$\calZ(\CC G)$ to~$\CC$
            by the formula.
        \end{proposition}

        \begin{corollary}[{\cite{SCHNEIDER1990601}}]\label{cor:CharacterCriterion}
            A class function~$\chi$ is an irreducible character of a finite group~$G$ 
            if and only if all of the following are true.
            \begin{enumerate}
                \item $\chi(\one) > 0$.
                \item $\sum_{g \in G} \abs{\chi(g)}^2 = \abs G$.
                \item 
                $\sum_{\gamma} c_{\alpha\beta\gamma} \chi(\one)\chi(C_\gamma) = \chi(C_\alpha) \chi(C_\beta)$ 
                for all conjugacy classes~$C_\alpha,C_\beta,C_\gamma$
                where $c_{\alpha\beta\gamma}$ are the class algebra constants 
                defined by~\ref{eq:classAlgebraConstants}.
            \end{enumerate}
        \end{corollary}

        Since we have computed~$c_{\alpha\beta\gamma}$ in Step~iii,
        the three conditions of~\ref{cor:CharacterCriterion} are readily checked.
        We were unsure if the ordering of the conjugacy classes 
        in the table produced by GAP (the column ordering)
        was the same as our ordering.
        Since we have conjugacy class sizes and the order of the representative group elements,
        most of the columns are unambiguous,
        but there were 5 pairs of columns that were left undecided based on the size and order.
        We examined all $2^5$ permutations of columns and found exactly one that passes
        the homomorphism criteria for all rows in the table.
        We confirm that the table from GAP system~\cite{GAP4} is indeed the character table.
        Since the table consists of integers only,
        arithmetic stays entirely within the rational numbers;
        we do not have to implement any algebraic number (field extensions of~$\mathbb Q$)
        arithmetic routines.

        \begin{proof}[Proof of~\ref{prop:central-character}]
            We prove the converse first.
            Let $\rho$ be an irrep that gives a character~$\chi$.
            We have an algebra homomorphism
            \begin{equation}
                \omega : \calZ(\CC G) \ni x \mapsto \frac{\chi_\rho(x)}{\chi_\rho(\one)} \in \CC
            \end{equation}
            because for any~$x \in \calZ(\CC G)$, the represented operator~$\rho(x)$
            is proportional to the identity by Schur's lemma
            \begin{equation}\begin{split}
                \rho(x) = \frac{\chi_\rho(x)}{\chi_\rho(\one)} \rho(\one) \, ,          
            \end{split}\end{equation}
            implying that for all~$x,y \in \calZ(\CC G)$ it holds that
            \begin{equation}
                \frac{\chi_\rho(xy)}{\chi_\rho(\one)} \rho(\one) 
                = \rho(xy) = \rho(x)\rho(y) 
                = \frac{\chi_\rho(x)}{\chi_\rho(\one)} \frac{\chi_\rho(y)}{\chi_\rho(\one)} \rho(\one) \, .
            \end{equation}
            
            For the forward assertion,
            we recall that the class sums generate the center~$\calZ(\CC G)$ of~$\CC G$, 
            a commutative $*$-algebra of dimension~$r$.
            A finite-dimensional commutative semisimple algebra 
            is always isomorphic to a direct product ring~$\CC^r$
            where $r$ is the dimension of the algebra.
            In particular, such a product ring has exactly $r$ unital algebra homomorphisms into~$\CC$.
            Since there are exactly $r$ characters, 
            the assertion is proved.
        \end{proof}

        \begin{proof}[Proof of~\ref{cor:CharacterCriterion}]
            The forward implication is obvious from~\ref{prop:central-character} and~\ref{eq:classAlgebraConstants}.
            For the backward implication,
            we define $\omega = \chi(\one)^{-1}\chi$, linearly extended to~$\calZ(\CC G)$.
            We show that $\omega$ is an algebra homomorphism $\calZ(\CC G) \to \CC$.
            The unitality and $\CC$-linearity holds by definition.
            The multiplicative property of~$\omega$ should be checked 
            for a basis of~$\calZ(\CC G)$.
            So, it suffices to check 
            $\omega(C_\alpha \cdot C_\beta) = \omega(C_\alpha) \omega(C_\beta)$
            for all $\alpha$ and $\beta$,
            which amounts to 
            $\sum_{\gamma} c_{\alpha\beta\gamma} \omega(C_\gamma) = \omega(C_\alpha) \omega(C_\beta)$.
            This is true by~(c).
            We conclude that $\omega$ is a unital homomorphism.
            Condition~(b) fixes $\chi(\one)$
            up to a unimodular complex number,
            which is fixed by~(a). 
            Now, \ref{prop:central-character} guarantees that~$\chi$ is a character.
        \end{proof}
        
        \vspace{1ex}
        \item[Step v.]\emph{Embedding of $F$ into $G$}.
        
        Enumerate all elements of~$\Sp(4;\FF_2) \times \Sp(2; \FF_2)$ by the same procedure as for~$G=\Sp(6;\FF_2)$.
        The number of elements is found to be~$(2^4-2^0)(2^3-2^1)2^2 2^1 \times 6 = 4320$.
        Embed $\Sp(4;\FF_2) \times \Sp(2; \FF_2)$ to obtain a list of all elements of~$F_{1,2} \le G$.
        This list has the same data structure as the list of elements of~$G$.
        Conjugate~$F_{1,2}$ by an involution corresponding to the swap of qubit~$1$ and $3$
        to obtain a list of all elements of~$F_{2,3}$.
        Verify that they form groups and perform other consistency checks.

        \vspace{1ex}
        \item[Step vi.]\emph{The five abstract irreps}.
        
        If $\rho$ does not contain any trivial irrep of~$F$ (multiplicity zero),
        then the projector~$\Pi(\rho|F_{1,2})$ is zero
        and the claim on the norm is vacuously true.
        The criterion for a positive multiplicity is
        \begin{equation}
            \frac 1 {|F|} \sum_{f \in F} \overline{\chi_\one(f)}  \chi_\rho(f) 
            = \frac 1 {|F|} \sum_{f \in F} \chi_\rho(f)> 0 \label{eq:H-trivial-multiplicity}
        \end{equation}
        where $\chi_\one(f) = 1$ is the character of the trivial irrep of~$F$.
        It turns out that all positive multiplicity irreps of~$G$ have dimension at most~$168$.
        
        Using the verified character table of~$G$ from Step~iv
        and the embedded subgroup elements from Step~v,
        we evaluate~\ref{eq:H-trivial-multiplicity} 
        by summing over all~$\abs{F_{1,2}} = \abs{F_{2,3}} = 4320$ elements.
        Report those irreps of~$G$ that have nonzero multiplicity.
        There are exactly five such $G$-irreps including the trivial irrep.

        \vspace{1ex}
        \item[Step vii.]\emph{Gelfand pairs}.
        
        This step only requires that $G$ is a finite group and $F$ is a subgroup.
        Let $\rho$ be a $G$-irrep and $\one$ the $F$-trivial representation.
        The representation of $G$ induced by $\one$ is denoted $\Ind_F^G\; \one$ and the restriction of $\rho$ to $F$ is denoted $\Res_F^G\; \rho$.
        \begin{definition}[{\cite[e.g.][Chapter~45]{bump2004lie}}]\label{def:GelfandPair}
            $(G,F)$ is called a Gelfand pair if $\Ind_F^G\; \one$ is multiplicity-free.
        \end{definition}
        Verify that $(\Sp(6; \FF_2), \Sp(4; \FF_2) \times \Sp(2; \FF_2))$ is a Gelfand pair by asserting the averages from \ref{eq:H-trivial-multiplicity} are equal to zero or one for all irreps.
        \begin{proposition}[{\cite[e.g.][Chapter~3.F]{diaconis1988group}}]\label{prop:GelfandPair}
        The following statements are equivalent.
            \begin{enumerate}
                \item $(G,F)$ is a Gelfand pair.
                \item For every irrep of $G$, its $F$-fixed subspace has dimension at most one.
            \end{enumerate}
        \end{proposition}
        \begin{proof}
            We show that (a) implies (b).
            By Frobenius reciprocity~\cite[e.g.][Corollary~3.20]{fulton2013representation}, the number of times that $\rho$ appears in $\Ind_F^G\; \one$ is equal to the number of times that $\one$ appears in $\Res_F^G\; \rho$.
            With the assumption that $(G,F)$ is a Gelfand pair, this implies that the number of times that $\one$ appears in $\Res_F^G\; \rho$ is at most one.
            Since $\one$ has dimension one, the dimension of the $F$-fixed subspace of $\rho$ is at most one.
            The reverse implication follows similarly.
        \end{proof}
        
        \vspace{1ex}
        \item[Step viii.]\emph{Sum of rank one projectors}.

        Denote the orthogonal projectors $P=\Pi(\rho|F_{1,2})$ and $Q=\Pi(\rho|F_{2,3})$.
        Suppose that $\rho$ contains a trivial irrep of~$F$.
        The norm of the sum of these projectors is
        \begin{equation}\label{eq:NormOfSumOfTwoProjectors}
            \norm{P + Q} = 1 + \norm{P Q}
            = 1 + \sqrt{\lambda_{\text{max}} \left( QPQ \right)}.
        \end{equation}
        Calculate the eigenvalues of $QPQ$ by trace powers.
        Since $(G, F)$ is a Gelfand pair, $P$ has rank one so $QPQ$ has rank one.
        Recall that $F_{2,3}$ is obtained from conjugating $F_{1,2}$
        by an involution $g = g^{-1}$,
        \begin{equation}
            Q = \rho(g) P \rho(g).
        \end{equation}
        Let $T = P \rho(g) P$.
        Then, $T^2 \simeq QPQ$ via conjugation by $\rho(g)$.
        The eigenvalue simplifies to
        \begin{equation}
            \lambda_{\text{max}}(QPQ) = (\Tr T)^2.
        \end{equation}
        The function $\phi_\rho(g) = \Tr T$ is called a spherical function.
        From the projector formula,
        \begin{equation}\label{eq:spherical-any}
        \begin{split}
            \phi_\rho(g) &= \Tr(P \rho(g))\\
            &= \Tr \left(\frac{1}{\abs{F_{1,2}}} \left(\sum_{f\in F_{1,2}} \rho(f)\right) \rho(g) \right)\\
            &= \frac{1}{\abs{F_{1,2}}} \sum_{f\in F_{1,2}} \chi_\rho(fg).
        \end{split}
        \end{equation}
        We conclude that
        \begin{equation}
            \norm{P + Q} = 1 + \abs{\phi_\rho(g)}.
        \end{equation}
        Calculate the spherical function $\phi_\rho(g)$ from the verified character table.
        Verify that this eigenspectrum numerically matches that of the permutation representation
        $\CC[G/F] \cong \Ind_F^G\; \one$, which has dimension 336.
        \qedhere
    \end{enumerate}
\end{proof}

\begin{remark}\label{rem:gapCalculationAI}
    We used Claude for the implementation.
    We wrote an initial prompt of the to-do list based on the above list of steps,
    and the prompt itself is refined and expanded by Claude Fable~5.
    The refined instruction was given to Claude Code (Opus 5)
    to produce the implementation, which we revised.
    We followed similar workflow for the orthogonal case below.
\end{remark}

\subsection{Orthogonal subgroups}

Each part above transfers to the orthogonal group $\SO(2^n)$.
The (special) orthogonal Pauli group is $\calP(n; \RR) = \langle \calP_n, \ii \one_{2^n} \rangle \cap \SO(2^n)$.
The (special) orthogonal Clifford group $\Cl(n;\RR)$ is the normalizer in $\SO(2^n)$ of the (special) orthogonal Pauli group.
The local orthogonal Clifford group~$\LCl(S,n;\RR)$ is a subgroup of~$\Cl(n;\RR)$
consisting of those supported on a subset~$S$ of qubits;
as before, each local group includes the full center of the parent group.

\begin{proposition}[\cite{nebe2000invariantscliffordgroups}]
\label{prop:OrthogonalCliffordQuotientPauli}
Let $n \ge 3$.
The quotient group of the orthogonal Clifford group $\Cl(n;\RR)$
by the orthogonal Pauli group $\calP(n; \RR)$ and the center
is isomorphic to the split orthogonal group $\OO^+(2n; \FF_2)$.
\end{proposition}

\begin{proof}
We present an elementary proof following~\cite[\S II]{Haah_2017},
which gives the analogous proposition for the unitary group.
Recall that the phase-forgotten Pauli group is the additive group~$\FF_2^{2n}$,
which happens to be a vector space.
The commutation data between two Pauli operators is captured by a symplectic form on their corresponding $2n$-bit vectors,
\begin{equation}
    \FF_2 \ni \lambda \left( \begin{bmatrix} a\\ b \end{bmatrix}, \begin{bmatrix} c\\ d \end{bmatrix}\right)
    =
    \begin{bmatrix} a & b \end{bmatrix}
    \begin{bmatrix} 0 & I\\ -I & 0 \end{bmatrix}
    \begin{bmatrix} c\\ d \end{bmatrix}
    = a\cdot d - b\cdot c
    = \begin{cases}
        0 & \text{they commute}\\
        1 & \text{they anticommute}.
    \end{cases}
\end{equation}
A Pauli operator $P$ is distinguished by whether it has an even or odd number of $\sigma_y$ tensor factors.
If this number is even, $P$ has real entries so it is orthogonal.
Otherwise, $\ii P$ has real entries so it is orthogonal.
The orthogonal Pauli group is the union of the even Pauli operators and $\ii$ times the odd Pauli operators.
The number of $\sigma_y$ tensor factors is captured by a quadratic form on $P$'s corresponding $2n$-bit vector,
\begin{equation}
    \phi \left( \begin{bmatrix} a\\ b \end{bmatrix}\right)
    =
    \begin{bmatrix} a & b \end{bmatrix}
    \underbrace{\begin{bmatrix} 0 & I\\ 0 & 0 \end{bmatrix}}_{\phi_n}
    \begin{bmatrix} a\\ b \end{bmatrix}
    = a\cdot b
    =
    \begin{cases}
        0 & P \text{ is even}\\
        1 & P \text{ is odd}.
    \end{cases}
\end{equation}
This quadratic form is a quadratic refinement of the symplectic form, $\lambda(x,y) = \phi(x+y) + \phi(x) + \phi(y)$.
The automorphism group of the quadratic form on $\FF_2^{2n}$ is contained in that of the symplectic form.
We claim that the automorphism group of the quadratic form is generated by
\begin{align}
    H = \begin{bmatrix}
        0 & 1\\
        -1 & 0
    \end{bmatrix},
    &&
    C_{i,j,i+n,j+n} = \begin{bmatrix}
        1 & 0 & &\\
        1 & 1 & &\\
        & & 1 & -1\\
        & & 0 & 1
    \end{bmatrix}
    &\text{ where } 1\le i\ne j \le n.
\end{align}
This claim is equivalent to transforming the matrix of any such automorphism $\Sigma$ into the identity,
using only $H$ and $C$~\cite[\S II.\,B]{Haah_2017}.
$H$ interchanges rows $i$ and $i+n$.
$C$ adds row $i$ to row $j$ ($1\le i\ne j \le n$) and row $j+n$ to row $i+n$.
Use $C$'s to reduce the upper half of the first column to
\begin{equation}
    \Sigma' =
    \begin{bmatrix}
        1 & \star_{1\times (2n-1)}\\
        0_{(n-1)\times 1} & \star_{(n-1)\times (2n-1)}\\
        \cmidrule(l{4pt}r{4pt}){1-2}
        a & \star_{1\times (2n-1)}\\
        b_{(n-1)\times 1} & \star_{(n-1)\times (2n-1)}
    \end{bmatrix}.
\end{equation}
Then, $H$'s on all but the first and $(n+1)$th row brings the submatrix $b$ to the upper half block.
Use $C$'s again to annihilate $b$,
\begin{equation}
    \Sigma'' =
    \begin{bmatrix}
        1 & \star_{1\times (2n-1)}\\
        0_{(n-1)\times 1} & \star_{(n-1)\times (2n-1)}\\
        \cmidrule(l{4pt}r{4pt}){1-2}
        a & c_{1\times (2n-1)}\\
        0_{(n-1)\times 1} & \star_{(n-1)\times (2n-1)}
    \end{bmatrix}.
\end{equation}
Preserving the quadratic form is equivalent to $\Sigma''^T \phi_n \Sigma'' - \phi_n$ is symmetric and zero on the diagonal.
This implies that $a=0$ and $c=\begin{bmatrix} 0_{(n-1)} & 1 & 0_{(n-1)} \end{bmatrix}$.
Repeat this procedure to reduce column $n+1$ to $e_{n+1}$ while preserving the first column,
\begin{equation}
    \Sigma''' =
    \begin{bmatrix}
        1 & 0_{1\times (n-1)} & 0 & 0_{1\times (n-1)}\\
        0_{(n-1)\times 1} & \star_{(n-1)\times (n-1)} & 0_{(n-1)\times 1} & \star_{(n-1)\times (n-1)}\\
        \cmidrule(l{4pt}r{4pt}){1-4}
        0 & 0_{1\times (n-1)} & 1 & 0_{1\times (n-1)}\\
        0_{(n-1)\times 1} & \star_{(n-1)\times (n-1)} & 0_{(n-1)\times 1} & \star_{(n-1)\times (n-1)}
    \end{bmatrix}.
\end{equation}
The desired transformation follows by recursion on the matrix whose quadrants are their respective $\star_{(n-1)\times (n-1)}$.

There exists a surjective group homomorphism $\psi$ from $\Cl(n;\RR)$ to the automorphism group of the quadratic form on $\FF_2^{2n}$.
This can be verified by checking that the generators of this automorphism group are induced (by conjugation) from two orthogonal Clifford operators~\cite[\S II.\,D]{Haah_2017},
\begin{align}
\label{eq:GeneratorsForOrthogonalClifford}
    U_H = \frac{1}{\sqrt{2}}\begin{bmatrix} 1 & 1\\ 1 & -1 \end{bmatrix} \in \mathsf O(2;\CC)\, ,
    &&
    U_{CNOT} = \begin{bmatrix}
        1 & 0 & &\\
        0 & 1 & &\\
        & & 0 & 1\\
        & & 1 & 0
    \end{bmatrix} \in \mathsf O(4;\CC).
\end{align}
We claim that the kernel of this homomorphism $\psi$ is the orthogonal Pauli group~\cite[\S II.\,D]{Haah_2017}.
Suppose that $U\in \ker \psi$, then its image is $\id$ so its induced action is $UP(v)U^\dagger = \pm P(v), v\in \FF_2^{2n}$.
This $U$ is uniquely determined (up to a sign) by its action on the generators of $\calP(n;\RR)$, which we take as the single-qubit Pauli operators $X_i$ and $Z_i$, where $i$ indexes the qubit.
On a single qubit, the equations $UXU^\dagger = \pm X$ and $UZU^\dagger = \pm Z$ has a unique solution for a choice of signs, $\{I,X,Z,XZ\}$.
This is extended to any $n$ by assuming $U$ is a tensor product of single-qubit operators (by uniqueness) and treating each tensor factor with the single qubit case.

We have proved that $\Cl(n;\RR) / \calP(n;\RR)$
is isomorphic to the automorphism group of the quadratic form $\phi$ on $\FF_2^{2n}$.
Quadratic forms over a finite field are classified by the Witt group~\cite[\S2]{haah2025topological}.
Over $\FF_2$, there are two Witt classes and $\phi$ belongs to the trivial class.
The automorphism group of~$\phi$ on $\FF_2^{2n}$ is the $+$ type split orthogonal group $\OO^+(2n;\FF_2)$.
\end{proof}

\begin{proposition}
    The two-qubit local orthogonal Clifford groups $\LCl(2;\RR)$
    generate the $n$-qubit orthogonal Clifford group $\Cl(n;\RR)$ for $n \ge 3$.
\end{proposition}

\begin{proof}
    The orthogonal Clifford group is generated by the orthogonal Pauli group $\calP(n;\RR)$
    with the elements $U_{CNOT}$
    acting on pairs of qubits and $U_H$ on any (fixed) qubit
    by~\ref{eq:GeneratorsForOrthogonalClifford}.
    Hence, it suffices to generate $\calP(n;\RR)$ by the elements in~\ref{eq:GeneratorsForOrthogonalClifford}
    over various qubits.
    These elements generate $\calP(2;\RR)$
    as easily seen starting with an identity
    \begin{equation}
        ((\one_2 \otimes U_H)(U_{CNOT}))^4 = Z \otimes \one_2 \, .
    \end{equation}
    Then, $\calP(2;\RR)$ generates $\calP(n;\RR)$
    by arranging the single-qubit Pauli operators as tensor factors.
\end{proof}

\begin{corollary}\label{cor:3qubitCliffordOrthogonal}
    Consider 3 qubits $a,b,c$ 
    and 
    let $\LCl^\RR_S$ denote the local group~$\LCl(S,n;\RR)$ where $n \ge 3$. 
    For any representation~$\rho$ of~$\LCl^\RR_{a,b,c}$,
    we have
    \begin{equation}
        \norm*{
            \frac{
                \Pi(\rho|\Cl^\RR_{ab}) + \Pi(\rho|\Cl^\RR_{bc})
            }{2}
            -
            \Pi(\rho|\Cl^\RR_{abc})
        }
        \le 1 - \frac {5}{18} \, .
    \end{equation}  
\end{corollary}

\begin{proof}
    This closely parallels the proofs of \ref{prop:sp6} and \ref{cor:3qubitClifford},
    so we describe only the differences.
    Similarly to~\ref{prop:sp6}, the calculation passes
    to the quotient group of the orthogonal Clifford group by the orthogonal Pauli group.
    This quotient group is isomorphic to the split orthogonal group $\OO^+(6; \FF_2)$ by \ref{prop:OrthogonalCliffordQuotientPauli}.
    
    The slightly larger generating set for $\OO^+(6;\FF_2)$ is $F = \OO^+(4;\FF_2)\times \OO^+(2;\FF_2)$.
    The gap calculation analogous to~\ref{prop:sp6}
    is complicated by the fact that
    $(\OO^+(6; \FF_2), \OO^+(4; \FF_2) \times \OO^+(2; \FF_2))$ is not a Gelfand pair;
    the induced representation $\Ind_F^G \one$ is multiplicity-free,
    except for one irrep which has multiplicity two.
    We calculate the spherical function $\phi_{\rho}(g)$ via~\ref{eq:spherical-any} 
    for the multiplicity-free irreps from the verified character table.
    The computational burden is minimal
    as the order of $F = \OO^+(4; \FF_2) \times \OO^+(2; \FF_2)$ is $144$
    and the order of $\OO^+(6;\FF_2)$ is $40\,320$.
    
    For the multiplicity-two irrep~$\rho$,
    we calculate the eigenvalues of $QPQ \simeq T^2$ by trace powers.
    Projector $P = \Pi(\rho|F)$ has rank two so $T = P \rho(g) P$ has rank at most two.
    The eigenvalues $\lambda_+, \lambda_-$ of~$T$ are the solutions to
    \begin{align}
    \begin{cases}
        \alpha = \Tr(T) = \lambda_+ + \lambda_-\\
        \beta = \Tr(T^2) = \lambda_+^2 + \lambda_-^2
    \end{cases}
    & \Longleftrightarrow \lambda_\pm = \frac{\alpha \pm \sqrt{2\beta - \alpha^2}}{2} \, .
    \end{align}
    We calculate $\phi_\rho(g) = \Tr(T)$ via \ref{eq:spherical-any}.
    We calculate $\Tr(T^2)$ from the projector formula
    \begin{equation}
    \label{eq:TracePowers}
        \Tr(T^2) = \Tr\big( (P\rho(g))^2 \big)
        = \Tr \left(\frac{1}{\abs{F}^2} \left(\sum_{h\in F_{1,2}} \rho(hg)\right)^2 \right)
        = \frac 1 {\abs{F}^2} \sum_{h,k\in F_{1,2}} \chi_\rho(hgkg)\, .
    \end{equation}
    In fact, the eigenspectrum of 
    $\Pi(\rho_{reg}|F_{1,2}) + \Pi(\rho_{reg}|F_{2,3})$
    for the regular representation~$\rho_{reg}$ 
    is $\{0,\frac{20}{36},\frac{35-\sqrt{193}}{36},\frac{37-\sqrt{193}}{36},\frac{24}{36},\frac{28}{36},\frac{32}{36},\frac{34}{36},\frac{38}{36},\frac{40}{36},\frac{44}{36},\frac{48}{36},\frac{35+\sqrt{193}}{36},\frac{37+\sqrt{193}}{36},\frac{52}{36},2\}$.
\end{proof}

\begin{theorem}\label{thm:BROCC}
    The spectral gap of the Brickwork Random Orthogonal Clifford Circuit of depth~$2$
    on (the one-dimensional chain of) $n$ qubits for any $n \ge 3$
    is 
    \begin{equation}
        \Delta(\nu_{\mathsf{BROCC}}, \Cl(n; \RR)) \ge \frac 1 {144} \, .
    \end{equation}
\end{theorem}
\begin{proof}
This exactly parallels the proofs of \ref{thm:additiveClifford} and \ref{thm:BRCC}.
\end{proof}

\section{Random Pauli Rotations}
\label{sec:RandomPauliRotations}

Let $\calP_n = \{\sigma^x, \sigma^y, \sigma^z, \one_2\}^{\otimes n} \setminus \{\one_{2^n}\}$
denote the set of nonidentity Pauli operators where $n \ge 1$.
We fix the overall signs for the Pauli operators 
so that there are exactly $4^n - 1$ elements in~$\calP_n$ and they are all self-adjoint.
They form a basis of the (complexified) Lie algebra~$\su(2^n)$ 
that is orthogonal with respect to the Killing form.

Let $\Theta \le \RR/2\pi\ZZ$ be a closed additive group.
We will use two examples of $\Theta$:
the full circle group~$\RR/2\pi \ZZ$ and 
a discrete group~$\{m \frac \pi 4 + 2\pi \ZZ~:~ m \in \ZZ\}$.
Since $\Theta$ is closed, we have a Haar probability measure~$\mu_\Theta$.
For any representation~$\rho$ of~$\SU(2^n)$ and $P \in \calP_n$,
we define an orthogonal projector
\begin{equation}
    \Pi(\rho|\Theta,P) 
    = \avg_{\theta \sim \mu_\Theta} \rho(e^{\ii \theta P}) 
    = \int \rho(e^{\ii \theta P}) ~\rd \mu_\Theta(\theta) \, .
\end{equation}
Define a real number
\begin{equation}
    \Delta_\Theta \in [0,3]
\end{equation}
to be the supremum real number 
such that for \emph{every} representation~$\rho$ of
the closed Lie or finite group $\mathsf L_\Theta \le \LU(2)$ generated 
by~$\{ e^{\ii \theta \sigma^x}, e^{\ii \theta \sigma^y}, e^{\ii \theta \sigma^z} ~:~ \theta \in \Theta\}
\subset \SU(2)$
(along with $\LCl(\emptyset,n)$ if $\Theta$ is discrete),
\begin{equation}
    \calH_0^2 \succeq \Delta_\Theta \calH_0 \qquad \text{where} \qquad 
    \calH_0 = \sum_{i \in \{x,y,z\}} \one - \Pi(\rho|\Theta,\sigma^i) \, .
\end{equation}
In other words, $\Delta_\Theta/3$ is the spectral gap of the random walk on~$\mathsf L_\Theta$
where we choose one~$\sigma^j$ of~$\sigma^x, \sigma^y,\sigma^z$
(and a central element~$c$ if $\Theta$ is discrete) uniformly at random
and apply $(c)e^{\ii \theta \sigma^j}$ where $\theta \sim \mu_\Theta$.

\begin{lemma}\label{lem:RandomPauliRotation}
    For any representation~$\rho$ of the closed group~$\mathsf G_\Theta$
    generated by $e^{\ii \theta P}$ with $P \in \calP_n$ and $\theta \in \Theta$
    (and the center $\LCl(\emptyset,n)$ if $\Theta$ is discrete),
    we have 
    \begin{equation}
        \norm[\Big]{
            \avg_{P \sim \calP_n} \Pi(\rho|\Theta,P) - \Pi(\rho|\mathsf G_\Theta)
        }
        \le
        1 - \frac{1 + 4^{n-1}(\Delta_\Theta - 1)}{4^n - 1} \, .
    \end{equation}
\end{lemma}

This proof is an adaptation of the triangle comparison 
for Kac's random walk~\cite[Theorem~1.1]{caputo2008spectral}.

\begin{proof}
    Define a ``Hamiltonian'' $\calH = \sum_{P \in \calP_n} (\one - \Pi(\rho|\Theta,P))$.
    Writing $F_P = \one - \Pi(\rho|\Theta,P)$ for $P \in \calP_n$,
    we consider the expansion of~$\calH^2$:
    \begin{align}
        \nonumber \calH^2 &= \underbrace{\sum_P F_P}_{\calH} + \sum_{\substack{P\ne Q\\ PQ = -QP}} F_P F_Q + \sum_{\substack{P\ne Q\\ PQ = QP}} F_P F_Q & \text{($F_P^2 = F_P$)}\\
        &\succeq \calH + \sum_{\substack{P\ne Q\\ PQ = -QP}} F_P F_Q & \text{($F_P F_Q = (F_P F_Q)^2 \succeq 0$ if $PQ=QP$)}.
    \end{align}
    Notice that a pair of anticommuting Pauli operators $P,Q$
    generate a subalgebra of~$\su(2^n)$ that is isomorphic to~$\su(2)$.
    Call $\{P, Q, R = -\ii P Q\} \subseteq \calP_n$ a $\su(2)$-triangle.
    Since $\rho$ restricts to a representation of the local group~$\mathsf L_\Theta$,
    we see that $\Delta_\Theta$ applies to every $\su(2)$-triangle:
    \begin{align}
        \nonumber (\underbrace{F_P + F_Q + F_R}_{\calH_T})^2 &\succeq \Delta_\Theta (F_P + F_Q + F_R)\, , \\
        \calH_T^2 - \calH_T &\succeq (\Delta_\Theta - 1) \calH_T \, .
    \end{align}
    Consider a simple graph~$\calG$ with a vertex set~$\calP_n$ 
    and an edge set consisting of anticommuting pairs.
    Recall that the phase-forgotten Pauli group is the additive group~$\FF_2^{2n}$,
    which happens to be a vector space, and
    given a Pauli operator~$P$, 
    a Pauli operator $Q$ that anticommutes with $P$
    if and only if the corresponding $2n$-bit vectors have nonzero symplectic product.
    This implies that $\calG$ is regular with degree~$2^{2n-1}$.
    An edge gives a unique $\su(2)$-triangle, 
    and hence different $\su(2)$-triangles do not share an edge.
    We see that
    \begin{equation}
        \sum_{\substack{P\ne Q\\ PQ=-QP}} F_P F_Q = 
        \sum_T (\calH_T^2 - \calH_T)
        \succeq
        (\Delta_\Theta - 1) \sum_T \calH_T
        = 4^{n - 1} (\Delta_\Theta - 1) \calH
    \end{equation}
    where the last equality is
    because every Pauli operator belongs to $4^{n -1}$ $\su(2)$-triangles.
    Therefore,
    \begin{equation}
        \calH^2 \succeq (1 + 4^{n - 1} (\Delta_\Theta - 1)) \calH \, .
    \end{equation}
    One minus the norm in the claim is 
    the spectral gap of~$\calH$ divided by~$\abs{\calP_n} = 4^n - 1$.
\end{proof}

\begin{lemma}
\label{lem:CliffordLocalGap}
    For $\Theta = \{ \frac{m \pi}{4} + 2\pi \ZZ ~:~m \in \ZZ\}$, a cyclic group of order~8,
    the unitary $e^{\ii \theta P}$ with any $\theta \in \Theta$ and $P \in \calP_n$
    is Clifford. Moreover,
    $
        \Delta_\Theta = 3/2.
    $
\end{lemma}

\begin{proof}
    The first claim follows by direct calculation that
    the conjugate $e^{\ii \pi P / 4} Q e^{-\ii \pi P / 4}$ of~$Q$ for any~$P,Q \in \calP_n$
    belongs to the Pauli \emph{group}~$\langle \calP_n, \ii \rangle$:
    since $e^{\ii \pi P / 4} = \one \cos \frac \pi 4 + \ii P \sin \frac \pi 4  = (\one + \ii P)/\sqrt{2}$,
    if $P$ and $Q$ anticommute,
    we see that $e^{\ii \pi P / 4} Q e^{-\ii \pi P / 4} = \ii P Q$ is a Pauli operator.

    The spectral gap calculation can be performed in a similar manner as in~$\Sp(6;\FF_2)$,
    but this case is much simpler.
    The group generated by~$\{ e^{\ii \theta P} ~|~ \theta \in \Theta, P \in \calP_1\}$
    together with the center~$\LCl(\emptyset,n)$ is
    is the group~$\LCl(1)$ of all single-qubit Clifford unitaries.
    Passing to the quotient group by moding out the center,
    we have a group of order~$24$ by~\ref{lem:factoringNormal}.
    Since $\SO(3) \cong \SU(2) / \{\pm \one\}$,
    we can understand the quotient group as a geometric symmetry group.
    
    Indeed, $e^{\ii \pi P / 4}$ corresponds to $90^\circ$ rotation about
    one of the coordinate axes of~$\RR^3$,
    the quotient group  is the symmetry group of the octahedron,
    which happens to be isomorphic to~$\symgp_4$, the symmetric group of degree~$4$.
    There are 6 elements of order~$4$ in~$\symgp_4$ in total,
    corresponding to~$e^{\pm \ii \pi P / 4}$ for $P=\sigma^x, \sigma^y, \sigma^z$.
    ($e^{\ii \pi P} \in \Cl(1)$ maps to $\one \in \SO(3)$.)
    One can easily confirm the spectral gap~$1/2$ 
    either in the regular representation or in all irreps of~$\symgp_4$.
\end{proof}

\begin{remark}
    Here is Mathematica code for the gap calculation in the regular representation.
    {\scriptsize   
    \begin{verbatim}
s4 = GroupElements[PermutationGroup[{Cycles[{{1, 2}}], Cycles[{{2, 3}}], Cycles[{{3, 4}}]}]];
PiX = {Cycles[{}], Cycles[{{1, 2, 3, 4}}], Cycles[{{1, 3}, {2, 4}}], Cycles[{{1, 4, 3, 2}}]};
PiY = {Cycles[{}], Cycles[{{1, 2, 4, 3}}], Cycles[{{1, 4}, {2, 3}}], Cycles[{{1, 3, 4, 2}}]};
PiZ = {Cycles[{}], Cycles[{{1, 3, 2, 4}}], Cycles[{{1, 2}, {3, 4}}], Cycles[{{1, 4, 2, 3}}]};
rep[cy_] := Transpose@Map[UnitVector[Length[s4],First@FirstPosition[s4, PermutationProduct[cy, #]]]&, s4]
Eigenvalues[(Total@Map[rep, PiX] + Total@Map[rep, PiY] + Total@Map[rep, PiZ])/12]
{1, 1/2, 1/2, 1/2, 1/2, 1/3, 1/3, 1/3, 1/3, 1/3, 1/3, 1/3, 1/3, 1/3, 0, 0, 0, 0, 0, 0, 0, 0, 0, 0}
\end{verbatim}
    }
\end{remark}

\begin{theorem}[implying \ref{thm:main}(2)]\label{thm:RPCR}
    The spectral gap of the Random Pauli Clifford Rotation on~$\Cl(n)$ is
    \begin{align*}
        \Delta(\nu_{\mathsf{RPCR}}, \Cl(n)) \ge \frac{4^n + 8}{8(4^n - 1)} > \frac{1}{8}.
    \end{align*}
\end{theorem}

\begin{proof}
    Apply \ref{lem:RandomPauliRotation} to $\Delta_\Theta = 3/2$ from \ref{lem:CliffordLocalGap}.
    Note that since $\Theta$ is discrete, the projectors average over the central phases too, $F_P = \one - \Ex_{c\sim \LCl(\emptyset,n)} \rho(c)\Pi(\rho|\Theta,P)$.
\end{proof}

\begin{lemma}[{\cite[Corollary 4.2]{haah2025efficient}}]
\label{lem:SULocalGap}
    For $\Theta = \RR/2\pi\ZZ$, we have
    $
        \Delta_\Theta = 5/4.
    $
\end{lemma}

\begin{proof}    
    In fact, the spectral gap 
    of Kac's random walk on~$\SO(3)$
    implies the lemma, too, which has appeared in~\cite{maslen2003eigenvalues,carlen2003determination,caputo2008spectral}.
    A way to see this is to realize that 
    in any even dimensional representation of~$\SU(2)$
    the projector~$\avg_{\theta \sim \RR/2\pi\ZZ} \rho(e^{\ii \theta P})$
    is equal to the kernel of the represented operator~$\rho_*(P)$ of element~$P \in \su(2)$,
    which is zero.
    Every odd dimensional representation of~$\SU(2)$ factors through~$\SO(3)$.
    Another way is to use a straightforward variant of~\ref{lem:factoringNormal}
    using $N_1 = N_2 = N_3 = \{\pm \one\} \le \SU(2)$, the center of~$\SU(2)$.
    Then, the spectral gap the random Pauli rotation on~$\SU(2)$ is calculated by
    passing to $\SU(2)/\{\pm \one\} \cong \SO(3)$.
    The Pauli matrices of~$\su(2)$ are represented in~$\so(3)$ as
    \begin{align*}
        \begin{bmatrix} 0 & 0 & 0\\ 0 & 0 & -1\\ 0 & 1 & 0 \end{bmatrix},\quad
        \begin{bmatrix} 0 & 0 & 1\\ 0 & 0 & 0\\ -1 & 0 & 0 \end{bmatrix},\quad
        \begin{bmatrix} 0 & -1 & 0\\ 1 & 0 & 0\\ 0 & 0 & 0 \end{bmatrix},
    \end{align*}
    which are precisely the generators of Kac's random walk.

    Yet another way is to reinterpret the proof of \cite[Lemma 4.1]{haah2025efficient} using the fact that $(\SU(2), \UU(1))$ is a Gelfand pair.
    They note that $g = \exp(-\ii \frac{\pi}{3\sqrt{3}}(\sigma^x + \sigma^y + \sigma^z))$ acts by conjugation on the projectors as $\Pi(\rho | \sigma^x) \mapsto \Pi(\rho | \sigma^y) \mapsto \Pi(\rho | \sigma^z) \mapsto \Pi(\rho | \sigma^x)$.
    The spherical function $\phi_\rho(g)$ is
    \begin{equation}
    \begin{split}
        \phi_\rho(g) &= \Tr \left(\Pi(\rho | \sigma^x) \rho(g)\right) = \Tr \left(\Pi(\rho | \sigma^y) \rho(g)\right) = \Tr \left(\Pi(\rho | \sigma^z) \rho(g)\right).
    \end{split}
    \end{equation}
    Index the nontrivial irreps by a positive integer $\ell$.
    These spherical functions are known to be related to the Legendre polynomials denoted by $P_\ell$~\cite[Section 3]{haagerup2016simple},
    \begin{equation}
        \phi_\ell(g) = P_\ell (\underbrace{2 \abs{g_{11}}^2 - 1}_0) =
        \begin{cases}
            (-1)^{\ell/2} 2^{-\ell}{\ell \choose \ell/2} & \text{if }\ell \text{ is even}\\
            0 & \text{if } \ell \text{ is odd}
        \end{cases}
    \end{equation}
    where $g_{11} = e^{-\ii \pi / 4}/ \sqrt{2}$ is the first entry of $g$ as a matrix in the basis that diagonalizes $\sigma^{\{x,y,z\}}$.
    Then, their proof of Corollary 4.2 is unchanged.
\end{proof}

\begin{remark}\label{rem:su4so6}
    The well known isomorphism between Lie algebras~$\su(4)$ and~$\so(6)$
    can be similarly used to give the spectral gap of the Random Pauli Rotation with~$n=2$ qubits
    because the spectral gap for Kac's random walk on~$\SO(m)$ is exactly known for 
    all~$m \ge 3$~\cite{maslen2003eigenvalues,carlen2003determination,caputo2008spectral},
    \begin{equation}
        \Delta(\nu_{\mathsf{RPR}}, \SU(2^2)) = \frac 2 {15}.
    \end{equation}
\end{remark}

\begin{theorem}[implying~\ref{thm:main}(1)]\label{thm:RPR}
The spectral gap of the Random Pauli Rotation on $\SU(2^n)$ is
    \begin{align*}
        \Delta(\nu_\mathsf{RPR}, \SU(2^n)) \ge \frac{4^n + 16}{16(4^n - 1)} > \frac 1 {16}.
    \end{align*}
\end{theorem}

\begin{proof}
    Apply \ref{lem:RandomPauliRotation} to $\Delta_\Theta = 5/4$ from \ref{lem:SULocalGap}.
\end{proof}

\begin{proposition}\label{lem:upper}
    When $n=1,2,3$, the lower bound in~\ref{thm:RPR} is tight.
    Furthermore, when $n \ge 3$,
    \begin{align*}
        \Delta(\nu_{\mathsf{RPR}}, \SU(2^n)) \le \frac{2^n(2^n-3)}{8(4^n-1)} < \frac{1}{8}.
    \end{align*}
\end{proposition}

\begin{proof}
    For $n=1$ and $n=2$, 
    the exact spectral gap can be extracted from the analysis of Kac's random walk;
    see~\ref{rem:su4so6}.
    For $n \ge 3$, we are going to calculate an eigenvalue of 
    the moment operator for the Random Pauli Rotation
    in a tensor representation
    \begin{equation}
        \tau_{4,4} : \SU(2^n) \ni U \mapsto U^{\otimes 4} \otimes \bar U^{\otimes 4} .
    \end{equation}
    This is naturally identified with the action of~$\SU(2^n)$ 
    on the space~$\End(\calH^{\otimes 4})$
    where $\calH = (\CC^2)^{\otimes n}$
    by
    \begin{equation}
        \tau_{4,4}(U) : \End(\calH^{\otimes 4}) \ni O \mapsto U^{\otimes 4} O U^{\dagger \otimes 4}  \in \End(\calH^{\otimes 4})\, ,
    \end{equation}
    with which
    the moment operator for~$\tau_{4,4}$ is a mixed unitary channel
    \begin{equation}
        \mathbf{\Phi} = M(\nu_{\mathsf{RPR}},\tau_{4,4},\SU(2^n)) : 
        O
        \mapsto \avg_{\substack{P\sim \calP_n\\\theta \sim \RR/2\pi\ZZ}}  
        (e^{\ii \theta P})^{\otimes 4} O (e^{-\ii \theta P})^{\otimes 4} \, .
    \end{equation}
    Since the eigenvalue of~$\mathbf{\Phi}$ of interest must not be~$1$,
    we must project out the trivial subrepresentation of~$\SU(2^n)$
    contained in~$\tau_{4,4}$.
    The orthogonal projector onto the trivial $\SU(2^n)$-subrepresentation is
    \begin{equation}
        \mathbf \Gamma = \Pi(\tau_{4,4}|\SU(2^n)) : \End(\calH^{\otimes 4}) \ni O \mapsto \int_{\SU(2^n)}  U^{\otimes 4} O U^{\dagger \otimes 4} \rd U
        \in \langle T_{\{1,2\}}, T_{\{2,3\}}, T_{\{3,4\}} \rangle
    \end{equation}
    where $T_B$ with $B \subset [4] = \{1,2,3,4\}$ with $\abs B = 2$ 
    are tensor factor transposition operators,
    which precisely generate the image of~$\mathbf{\Gamma}$
    by the Schur--Weyl duality~\cite[e.g.][Theorem $1.10$]{christandl2006structure}:
    \begin{equation}
        \calH^{\otimes 4} = \bigoplus_{\lambda \vdash 4} \Pi_\lambda \calH^{\otimes 4}
        \qquad \text{ with } \quad 
        \Pi_\lambda \calH^{\otimes 4} \cong \calQ_\lambda \otimes \calS_\lambda
    \end{equation}
    where $\calQ_\lambda$ is an irrep of~$\SU(2^n)$
    and $\calS_\lambda$ is an irrep of the symmetric group~$\symgp_4$
    associated with a Young diagram~$\lambda$.
    Note that three superoperators~$\mathbf{\Gamma}$, $\mathbf{\Phi}$, and 
    \begin{equation}
        \mathbf{\Pi}_\lambda : X \mapsto \Pi_\lambda X \Pi_\lambda
    \end{equation}
    commute pairwise.
    
    Define a distinguished element%
    \footnote{This $\Omega$ commutes with $U^{\otimes 4}$ for any~$U \in \Cl(n)$~\cite{bittel2025complete}.}
    of~$\End(\calH^{\otimes 4})$,
    \begin{equation}
        \Omega = \one_{2^n}^{\otimes 4} + \sum_{P \in \calP_n} P^{\otimes 4} = 4^{-n} \Omega^2 \, ,
    \end{equation}
    which commutes with $\Pi_\lambda \in \End(\calH^{\otimes 4})$ for every~$\lambda$,
    and its projected down version
    \begin{equation}
        E_\lambda = (\one - \mathbf{\Gamma})(\Pi_\lambda \Omega \Pi_\lambda) \, .
    \end{equation}    
    We claim that
    \begin{equation}
        \mathbf{\Phi}(E_\lambda) = e_\lambda E_\lambda
        \quad \text{ where $e_\lambda < 1$ and $E_\lambda \neq 0$ for ``valid''~$\lambda$.}
    \end{equation}
    Here, the condition on~$\lambda$ so as to make $E_\lambda \neq 0$ will be derived later.
    The difference $1 - e_\lambda$ for any valid~$\lambda$ 
    is an upper bound on~$\Delta(\nu_{\mathsf{RPR}}, \SU(2^n))$.
    
    To evaluate~$e_\lambda$, we calculate~$\mathbf{\Phi}(\Omega)$,
    starting with the average over~$\theta \sim \RR/2\pi \ZZ$.
    For anticommuting Pauli operators~$P$ and~$Q$,
    we have
    \begin{align}
        \avg_\theta (e^{\ii \theta P})^{\otimes 4} Q^{\otimes 4} (e^{-\ii \theta P})^{\otimes 4}
        &=
        \int_0^{2\pi} \frac{\rd \theta}{2\pi} ~ \left( 
            Q \cos 2\theta + \ii P Q \sin 2\theta
        \right)^{\otimes 4} \\
        &=
        \frac{3}{8}Q^{\otimes 4} + \frac{3}{8} (PQ)^{\otimes 4} 
        - \frac{1}{8} \sum_{\substack{B \subset [4] :\\ \abs B = 2}} P^{\otimes B} Q^{\otimes 4} \nonumber
    \end{align}
    where $[4] = \{1,2,3,4\}$ and
    $P^{\otimes B}$ is the tensor product of~$P$ and~$\one_{2^n}$
    with a tensor factor~$P$ for each position specified by an element of~$B$;
    for example, if $B = \{1,4\}$, then $P^{\otimes B} = P \otimes \one_{2^n} \otimes \one_{2^n} \otimes P$.
    For any~$P \in \calP_n$, we define
    \begin{equation}
        A_P = \sum_{Q \in \calP_n:\, PQ=-QP} Q^{\otimes 4} \, .
    \end{equation}
    Summing over~$Q \in \calP_n \cup \{\one\}$ and then averaging over~$P$, we have 
    \begin{align}
        \nonumber\avg_\theta (e^{\ii \theta P})^{\otimes 4} \Omega (e^{-\ii \theta P})^{\otimes 4}
        &= \sum_{\substack{Q \in \calP_n \cup \{\one\}:\\ PQ = QP}} Q^{\otimes 4}
        +
        \sum_{\substack{Q \in \calP_n :\\PQ = - QP}}
        \left(
        \frac{3}{8}Q^{\otimes 4} + \frac{3}{8} (PQ)^{\otimes 4} 
        - \frac{1}{8} \sum_{\substack{B \subset [4] :\\ \abs B = 2}} P^{\otimes B} Q^{\otimes 4} \right)\nonumber\\
        &=
        (\Omega - A_P) + \sum_{\substack{Q \in \calP_n :\\PQ = - QP}}
        \left(
            \frac{3}{4}Q^{\otimes 4} 
            - \frac{1}{8} \sum_{\substack{B \subset [4] :\\ \abs B = 2}} P^{\otimes B} Q^{\otimes 4} 
        \right)\\
        \nonumber &= 
        \Omega - \frac 1 4 A_P - \frac 1 8 \sum_{\substack{B \subset [4] :\\ \abs B = 2}} P^{\otimes B} A_P \, ,\\
        \label{eq:upperbound}
        \mathbf{\Phi}(\Omega) &= 
        \frac{1}{4^n-1} \sum_{P\in \calP_n}
        \left( 
            \Omega - \frac{1}{4} A_P - \frac{1}{8} \sum_{\substack{B \subset [4] :\\ \abs B = 2}} P^{\otimes B} A_P
        \right).
    \end{align}
    Since every nonidentity Pauli operator anticommutes with exactly $2\cdot 4^{n-1}$ Pauli operators,
    \begin{equation}
        \sum_{P\in \calP_n} A_P = 2 \cdot 4^{n-1} \sum_{P\in \calP_n} P^{\otimes 4} = 2 \cdot 4^{n-1} (\Omega - \one_{2^n}^{\otimes 4}).
    \end{equation}
    
    To simplify $\sum_P \sum_B P^{\otimes B} A_P$,
    we recall an identity 
    $
        \sum_{P\in \calP_n \cup \{\one\}} P^{\otimes 2} = 2^n T,
    $
    where $T \in \End(\calH^{\otimes 2})$ is the unitary that swaps the two tensor factors.
    This identity can be checked by considering the fact that Pauli operators
    form an orthogonal basis under the Hilbert--Schmidt inner product.
    Now, observe that for $P,Q \in \calP_n \cup \{\one\}$,
    if $PQ = QP$, then $PQ \in \calP_n \cup \{ \one\}$,
    but if $PQ = -QP$, then $PQ \notin \calP_n \cup \{\one\}$ because of an extra~$\ii$.
    In particular,
    \begin{align}
        2^n T (Q^{\otimes 2})
        &= \sum_{P \in \calP_n \cup \{\one\}: PQ = QP} ( P Q)^{\otimes 2} 
        - \sum_{P \in \calP_n: PQ = -QP} (\ii P Q)^{\otimes 2}\nonumber\\
        2^n T &=
        \sum_{P \in \calP_n \cup \{\one\}: PQ = QP} ( P Q)^{\otimes 2} 
        + \sum_{P \in \calP_n: PQ = -QP} (\ii P Q)^{\otimes 2}\, ,\nonumber\\
        2^{n-1}(T - T(Q^{\otimes 2})) 
        &=
        \sum_{P \in \calP_n: PQ = -QP} (\ii P Q)^{\otimes 2} = 
        \sum_{R \in \calP_n: RQ = -QR} R^{\otimes 2} \, .
    \end{align}
    Embedded into~$\End(\calH^{\otimes 4})$,
    this becomes 
    \begin{align}
        \sum_{P \in \calP_n: PQ =-QP} P^{\otimes B} &= 2^{n-1}(T_B - T_B Q^{\otimes B}), \label{eq:tbtbq}\\
        \sum_{P\in \calP_n \cup \{\one\}} P^{\otimes B} &= 2^n T_B, \label{eq:swapIdentity}
    \end{align}
    where $T_B$ is now a transposition of the two tensor factors specified by~$B \subset [4]$.

    Observe that the double sum~$\sum_{P \in \calP_n}\sum_{Q\in \calP_n:PQ=-QP}$
    is over all ordered pairs of~$P,Q$ that are anticommuting.
    So,
    \begin{align}
        \sum_{P\in \calP_n} \sum_{\substack{B\subset [4]:\\ |B|=2}} P^{\otimes B} A_P 
        &=  
        \sum_{\substack{B\subset [4]:\\ |B|=2}} \sum_{Q\in \calP_n} \left(\sum_{\substack{P\in \calP_n: PQ = -QP}} P^{\otimes B}\right) Q^{\otimes 4}\nonumber\\
        &=  2^{n-1} \sum_{\substack{B\subset [4]:\\ |B|=2}} 
        \sum_{Q\in \calP_n} \left( T_B Q^{\otimes 4} - T_B Q^{\otimes B^c} \right) &\text{by \ref{eq:tbtbq}}\\
        &= 2^{n-1} \sum_{\substack{B\subset [4]:\\ |B|=2}} \left(T_B (\Omega - \one_{2^n}^{\otimes 4}) - T_{B} (2^n T_{B^c} - \one_{2^n}^{\otimes 4})  \right) &\text{by \ref{eq:swapIdentity}}
        \nonumber\\
        &= 2^{n-1} \sum_{\substack{B\subset [4]:\\ |B|=2}} \left(T_B \Omega  - 2^n T_{B} T_{B^c}  \right) .
        \nonumber
    \end{align}
    Plugging these into \ref{eq:upperbound}, we have
    \begin{align}
        \mathbf{\Phi}(\Omega) = \Omega - \frac{4^n}{8(4^n-1)}(\Omega - \one_{2^n}^{\otimes 4}) - \frac{2^n}{16(4^n-1)} \sum_{\substack{B\subset [4]:\\ |B|=2}} \left(\Omega T_B - 2^n T_{B^c} T_B\right).
    \end{align}

    Now we apply $(\one - \mathbf{\Gamma})\mathbf{\Pi}_\lambda$ on both sides.
    Observe that $\one_{2^n}^{\otimes 4}$ and $T_{B^c} T_B$ are both
    in the image of~$\mathbf{\Gamma}$, so they are eliminated under~$\one - \mathbf{\Gamma}$.
    \begin{align}
        \nonumber(\mathbf{\Phi} \circ (\one - \mathbf{\Gamma}) \circ \mathbf{\Pi}_\lambda)(\Omega) 
        &= 
        (\one - \mathbf{\Gamma}) \big(\Pi_{\lambda} \mathbf{\Phi} (\Omega) \Pi_\lambda \big)\\
        &= 
        \left(1 - \frac{4^n}{8(4^n-1)}\right) E_\lambda 
        - \frac{2^n}{16(4^n-1)} ((\one - \mathbf{\Gamma})\circ\mathbf{\Pi}_{\lambda}) \left(\Omega \sum_{\substack{B\subset [4]:\\ |B|=2}} T_B\right).
    \end{align}
    The sum $\sum_{B\subset [4]: |B|=2} T_B$ is a class sum of~$\symgp_4$
    and hence commute with all represented operators of $\SU(2^n) \times \symgp_4$.
    When restricted to~$\Pi_\lambda \calH^{\otimes 4}$,
    by Schur's lemma,
    this must be a scalar~$\kappa_\lambda$ multiple of~$\Pi_\lambda$
    \begin{align}
        \kappa_\lambda \Pi_\lambda & = \sum_{B\subset [4]: |B|=2} T_B \Pi_\lambda \, ,\\
        \kappa_\lambda (\dim \calQ_\lambda)(\dim \calS_\lambda) &= 6 (\dim \calQ_\lambda) \chi_\lambda( (12) ) \, ,
        \nonumber
    \end{align}
    where $\chi_\lambda$ is the character of~$\symgp_4$ on~$\calS_\lambda$.
    Therefore, we finally arrive at
    \begin{align}
        \mathbf{\Phi}(E_\lambda) &= \left(1 - \frac{2\cdot 4^n + 2^n \kappa_\lambda}{16(4^n-1)}\right) E_\lambda.
    \end{align}
    The character table of~$\symgp_4$ is well known~\cite[e.g.][\S2.3]{fulton2013representation}:
    \begin{equation}
        \label{s4character}
        \begin{array}{r|ccc|c}
        \symgp_4
            & \lambda
            & \chi_\lambda((12))
            & \dim \calS_\lambda 
            & \kappa_\lambda = \frac{6 \chi_\lambda((12))}{\dim \calS_\lambda} \\
        \hline
        \text{trivial}
            & (4)
            & 1
            & 1 
            & 6\\
        \text{sign irrep}
            & (1^4)
            & -1
            & 1 
            & -6\\
        \text{standard}
            & (3,1)
            & 1
            & 3 
            & 2 \\
        \text{sign $\otimes$ standard}
            & (2,1,1)
            & -1
            & 3 
            & -2\\
        \text{two-dimensional irrep}
            & (2,2)
            & 0
            & 2
            & 0
        \end{array}
    \end{equation}
    The claimed bound for~$n \ge 3$ follows by choosing~$\lambda = (1^4)$ and $\kappa_\lambda = -6$,
    provided that~$E_\lambda \neq 0$ for all~$n \ge 3$.

    Hence, it remains to show that $E_\xi = (\one - \mathbf{\Gamma})(\Pi_\xi \Omega) \neq 0$
    where $\xi = (1^4)$.
    We calculate the Frobenius norms $\norm*{\Pi_\xi \Omega}_2$
    and $\norm*{\Pi_\xi \mathbf{\Gamma}(\Omega)}_2$
    and show that they are different.
    This will conclude the proof.

    For $\eta \vdash 4$, 
    let $C_\eta = \sum_{\sigma \in \symgp_4:\mathrm{type}(\sigma) = \eta} T_\sigma$
    be the class sum, 
    the sum of all represented operators (on~$\calH^{\otimes 4}$) 
    of~$\sigma$ that has cycle type~$\eta$;
    for example, $C_{(2,2)} = T_{1,2} T_{3,4} + T_{1,3} T_{2,4} + T_{1,4} T_{2,3}$.
    To avoid clutter, 
    we write $C_1 = C_{(1,1,1,1)} = \one_{2^n}^{\otimes 4}$,
    $C_2 = C_{(2,1,1)}$,
    $C_3 = C_{(3,1)}$,
    $C_4 = C_{(4)}$,
    and
    $C_{2,2} = C_{(2,2)}$.
    By the Young symmetrizer,
    we know
    \begin{equation}
        \Pi_\xi = \frac{C_{1} - C_{2} + C_{3} - C_4 + C_{2,2} }{24} \, .
    \end{equation}
    In addition, the Weingarten calculus tells us
    that for nonidentity Pauli~$P \in \calP_n$ with $d = 2^n \ge 8$
    \begin{align}
        \nonumber \avg_{U \sim \SU(2^n)} (U P U^\dagger)^{\otimes 4} 
        &= 
        \frac{3 C_1 - d C_2 + 3 C_3 - d C_4 + (d^2-6) C_{2,2}}{(d^2 - 1)(d^2-9)} \, ,\\
        \mathbf{\Gamma}(\Omega) &= 
        C_1 + \frac{3 C_1  - d C_2 + 3 C_3 - d C_4 + (d^2-6) C_{2,2}}{d^2-9}.
    \end{align}
    The remainder of the calculation is straightforward
    using the following table of traces.
    \begin{equation}
        \begin{array}{c|ccccc}
                     C_\eta         & C_1 & C_2 & C_3 & C_4 & C_{2,2} \\
                                    \hline
        \Tr(P^{\otimes 4} C_\eta) 
        \text{ with } 
        P \neq \one_{2^n}           & 0   & 0   & 0   & 6 d  & 3 d^2    \\
        \Tr(C_\eta)                 & d^4 & 6 d^3& 8 d^2& 6 d  & 3 d^2  
        \end{array}
    \end{equation}
    All calculation is done with the class sums:
    \begin{equation}
        \left(\sum_{\alpha} x_\alpha C_\alpha \right)
        \left(\sum_{\beta} y_\beta C_\beta \right) = \sum_{\alpha,\beta,\gamma} x_\alpha y_\beta c_{\alpha\beta\gamma} C_\gamma \, .
    \end{equation}
    After some algebra,\footnote{We use Wolfram Mathematica with 
    a multiplication function
    $(x_1,x_2,x_3,x_4,x_{2,2})(y_1,y_2,y_3,y_4,y_{2,2})
    = (3 x_{2,2} y_{2,2}+x_1 y_1+6 x_2 y_2+8 x_3 y_3+6 x_4 y_4,y_2 x_{2,2}+x_2 y_{2,2}+2 \left(y_4 x_{2,2}+x_4 y_{2,2}\right)+x_2 y_1+x_1 y_2+4 \left(x_3 y_2+x_2 y_3\right)+4 \left(x_4 y_3+x_3 y_4\right),3 \left(y_3 x_{2,2}+x_3 y_{2,2}\right)+x_3 y_1+3 x_2 y_2+x_1 y_3+4 x_3 y_3+3 x_4 y_4+3 \left(x_4 y_2+x_2 y_4\right),y_4 x_{2,2}+x_4 y_{2,2}+2 \left(y_2 x_{2,2}+x_2 y_{2,2}\right)+x_4 y_1+4 \left(x_3 y_2+x_2 y_3\right)+x_1 y_4+4 \left(x_4 y_3+x_3 y_4\right),y_1 x_{2,2}+x_1 y_{2,2}+2 x_{2,2} y_{2,2}+2 x_2 y_2+8 x_3 y_3+2 x_4 y_4+4 \left(x_4 y_2+x_2 y_4\right))$
    for the class algebra of~$\symgp_4$.
    }
    we have
    \begin{align}
        \norm*{\Pi_\xi \Omega}_2^2 &= 4^n \Tr(\Pi_\xi \Omega)
        = 
        \frac{d^4 (d-1)(d-2)}{6}
        \\
        \norm*{\Pi_\xi \mathbf{\Gamma}(\Omega)}_2^2 
        &= \Tr\big(\Pi_\xi \mathbf{\Gamma}(\Omega)^2\big)
        = 
        \frac{2 d^3(d-1)(d-2)}{3(d-3)} \, ,
    \end{align}
    which are different for all~$d > 4$.
    At $d=4$ they are the same, which is consistent with the fact that
    $\Pi_\xi$ projects onto the trivial determinant representation of~$\SU(4)$.
\end{proof}

\begin{remark}\label{rem:gapUpperBoundAIProof}
    We ran numerical experiments on $n=1,2,3$ with $\tau_{4,4}$ 
    to find the eigenvector with largest eigenvalue, 
    and then asked ChatGPT (GPT-$5.5$) to conjecture a symbolic eigenvector.
    With more prompting, ChatGPT came to the construction of~$E_\lambda$ 
    and suggested a proof.
    We revised it and coined the argument for~$E_\xi \neq 0$.
\end{remark}

\begin{conjecture}[Aldous-type]
When $n \ge 3$, the spectral gap of the Random Pauli Rotation on $\SU(2^n)$ is
    \begin{align*}
        \Delta(\nu_{\mathsf{RPR}}, \SU(2^n)) = \frac{2^n(2^n-3)}{8(4^n-1)}.
    \end{align*}
That is, the upper bound \ref{lem:upper} from $t=4$ is tight.
In particular, the tight $\SU(2^n)$ irrep is the one identified by the highest weight $(1^4,0^{2^n-8},-1^4)$.
\end{conjecture}

\begin{noteadded}
\cite{dong2026spectral} claims to resolve this conjecture.
With the projective unitary group $\PU(2^n)$, they claim that the conjectured gap and irrep are tight.
With the special unitary group $\SU(2^n)$, they claim that the gap is $\frac{2^n - 8}{8(2^n-1)}$ from tight irrep $\bigwedge^8 \CC^{2^n}$ when $n\ge 4$.
\end{noteadded}

\begin{remark}
    This conjecture is born from a tradition of Aldous-type conjectures.
    Around~1992, Aldous conjectured that the spectral gap of transpositions generating the symmetric group is contained in the standard representation, for any such probability measure on transpositions.
    The uniform measure case had been proved several years earlier by \cite{diaconis1981generating}.
    In 2009, \cite{caputo2010proof} proved the conjecture 
    using the ``octopus inequality.''
    Recently, \cite{alon2026aldous} proved a similar spectral gap result 
    for an analogous random walk on $\SU(2^n)$.
    However, this gap is exponentially small in the number of qubits 
    and the walk does not use the tensor product structure.
    For Kac's random walk on $\SO(m)$, the spectral gap is realized by a $t=2$
    subrepresentation~\cite{maslen2003eigenvalues,carlen2003determination,caputo2008spectral}.
    This conjectured $t=4$ phenomenon for random unitary circuits has been noted by \cite{haah2025efficient,chen2025incompressibility}.
\end{remark}

\subsection{Orthogonal subgroups}

The lower bound on the spectral gap of the Random Pauli Rotation~\ref{thm:RPR}
transfers to the orthogonal group.
Define the subset of Pauli operators with entries in $\ii \RR$,
\begin{equation}
    \mathcal Y_n = \{P\in \calP_n \mid P \text{ contains an odd number of $\sigma^y$ tensor factors} \}.
\end{equation}
We fix the overall signs for $\ii \mathcal Y_n$
so that there are exactly $2^n(2^n-1)/2$ elements and they are all skew-adjoint.
They form a basis of the (real) Lie algebra~$\so(2^n)$ 
that is orthogonal with respect to the Killing form.

The counting for the number of elements follows from
considering the phase-forgotten Pauli group as the additive group $\FF_2^{2n}$,
which happens to be a vector space.
A Pauli operator is represented as a vector $(x|z) \in \FF_2^{2n}$
where $x,z$ are bit strings of length $n$.
The number of $\sigma^y$ tensor factors is even if and only if $x\cdot z = 0 \in \FF_2$.
Fix $x$ and consider the linear functional $z\mapsto x \cdot z$.
If $x=0$, its kernel has $2^n$ elements;
otherwise, its kernel has $2^{n-1}$ elements.
The number of orthogonal Pauli operators is $4^n - (2^{n} + (2^n-1)\cdot 2^{n-1}) = 2^n (2^n-1)/2$.

\begin{theorem}\label{thm:RPOR}
    The spectral gap of the Random Pauli Orthogonal Rotation on $\SO(2^n)$ is
    \begin{align*}
        \Delta(\nu_\mathsf{RPOR}, \SO(2^n)) \ge \frac{4^{n-1}-2^{n-1}+8}{4^{n+1}-2^{n+2}} > \frac 1 {16}.
    \end{align*}
\end{theorem}

\begin{proof}
    This closely parallels the proofs of \ref{lem:RandomPauliRotation} and \ref{lem:SULocalGap},
    so we describe only the differences.
    The proof of \ref{lem:RandomPauliRotation} uses three facts
    about the Pauli operators as a basis for $\su(2^n)$:
    (i) they are closed under commutators,
    (ii) the anticommutation graph is regular with high degree, and
    (iii) an edge gives a unique $\su(2)$-triangle.
    For $\mathcal Y_n$ as a basis for $\so(2^n)$,
    we use $\so(3)$-triangles generated by an anticommuting pair.
    The first and third facts clearly remain true.
    Since the orthogonal Clifford group acts transitively on~$\mathcal Y_n$,
    the anticommutation graph on~$\mathcal Y_n$ is regular.
    The degree is $2^{n-1}(2^{n-1}-1)$ and
    every vertex belongs to $2^{n-2}(2^{n-1}-1)$ $\so(3)$-triangles.
    Therefore, for any representation~$\rho$ of~$\mathsf G_\Theta$, we have 
    \begin{equation}
        \norm[\Big]{
            \avg_{P \sim \mathcal Y_n} \Pi(\rho|\Theta,P) - \Pi(\rho|\mathsf G_\Theta)
        }
        \le
        1 - \frac{1 + 2^{n-2}(2^{n-1}-1)(\Delta_\Theta - 1)}{2^n(2^n-1)/2} \, .
    \end{equation}
    The local gap~$\Delta_\Theta$ from \ref{lem:SULocalGap} remains the same
    since $\SU(2)/\{\pm \one\} \cong \SO(3)$.
\end{proof}

\section{Brickwork Random Unitary and Orthogonal Circuits}
\label{sec:BrickworkRandomUnitary}
\begin{theorem}\label{thm:AdditiveRandomUnitary}
    We identify $\SU(2^n)$ with the unitary group on $n$ qubits 
    labeled by $\{ 1,2,\ldots, n \}$.
    Let $\LU(\{i,j\}) =  \LU(\{i,j\},n) \le \SU(2^n)$
    be the local unitary group on two qubits~$i \neq j$,
    which includes the center of~$\SU(2^n)$.
    Let $\rho$ be any nontrivial irrep of~$\SU(2^n)$.
    Then,
    \begin{align}
        \norm*{\frac{1}{n -1} \sum_{i=1}^{n -1}\avg_{h \sim \LU(\{i,i+1\},n)} \rho(h)}
        \le 1 - \frac{1}{9 \cdot 2^{45}(n-1)} \, . \label{eq:additive1DRandomUnitary}\\
        \norm*{\frac{2}{n(n -1)} \sum_{i < j} \avg_{h \sim \LU(\{i,j\},n)} \rho(h)}
        \le 1 - \frac{1}{9 \cdot 2^{45}(n-1)} \, . \label{eq:all-to-all-RandomUnitary}
    \end{align}
\end{theorem}

\begin{proof}
    As in~\ref{eq:JensenForAllToAll},
    the left-hand side of~\ref{eq:all-to-all-RandomUnitary}
    is upper-bounded by that of~\ref{eq:additive1DRandomUnitary}.
    We do not repeat the argument here.

    Recall that the distribution~$\nu_{\mathsf{RPR}}$ of random Pauli rotations $e^{\ii \theta P}$ 
    with $\theta \sim \RR/2\pi\ZZ$ and $P \sim \calP_n$
    is the distribution
    \begin{equation}
        U e^{\ii \theta \sigma^z_1} U^\dagger \quad \text{ where } 
        U \sim \Cl(n)\,\, ,\,\, \theta \sim \RR / 2\pi\ZZ \, ,
    \end{equation}
    where $\sigma^z_1$ is the Pauli~$\sigma^z$ on qubit~$1$,
    because $U \sigma^z_1 U^\dagger$ with $U \sim \Cl(n)$
    is uniformly distributed on $\langle \pm \one, \calP_n \rangle$.
    We therefore have an equation of projectors:
    \begin{equation}
        \left(\underbrace{\avg_{U \sim \Cl(n)} \rho(U)}_A \right) 
        \left(\underbrace{\avg_{\theta \sim \RR/2\pi\ZZ} \rho(e^{\ii \theta \sigma^z_1})}_B \right)
        \left(\underbrace{\avg_{U \sim \Cl(n)} \rho(U)}_A \right) 
        =
        \left(\underbrace{\avg_{U \sim \Cl(n)} \rho(U)}_A \right) 
        \left(\avg_{W \sim \nu_{\mathsf{RPR}}} \rho(W) \right) .       
    \end{equation}
    Taking the norm, we see that
    \begin{align}
        \nonumber\norm*{
            ABA
        }
        &=
        \norm[\Big]{
            A \cdot \avg_{W \sim \nu_{\mathsf{RPR}}} \rho(W)
        }\\
        \nonumber&\le
        \norm{A} \cdot 
        \norm[\Big]{
            \avg_{W \sim \nu_{\mathsf{RPR}}} \rho(W)
        }\\
        &\le
        \norm[\Big]{
            \avg_{W \sim \nu_{\mathsf{RPR}}} \rho(W)
        } & \text{($A$ is a projector)}\\
        \nonumber&< 1 - \frac{1}{2^4} &\text{(by~\ref{thm:RPR})}\, .
    \end{align}
    Recall that $\nu_{\mathsf{BRCC}}$ denotes the probability distribution on~$\Cl(n)$
    induced by the Brickwork Random Clifford Circuit of depth~$2$.
    From~\ref{thm:BRCC}, we set a (non-optimal, comfortable) constant
    \begin{equation}
        m \ge 2^{10}
    \end{equation}
    to use~\ref{remark:RandomWalks} and see
    \begin{equation}
        \norm*{
            \underbrace{\avg_{U \sim \nu_{\mathsf{BRCC}}^{*m}} \rho(U)}_C - \underbrace{\avg_{W \sim \Cl(n)} \rho(W)}_A
        }
        \le \Big(1 - \Delta(\nu_{\mathsf{BRCC}}, \Cl(n)) \Big)^m < \frac 1 {2^6} \, .
    \end{equation}
    We are going to consider a distribution
    \begin{equation}
        \convo = \nu_{\mathsf{BRCC}}^{*m} * \beta * \nu_{\mathsf{BRCC}}^{*m}
    \end{equation}
    on~$\SU(2^n)$
    where $\beta$ is the Haar probability measure on~$H_0 = \{ e^{\ii \theta \sigma^z_1} : \theta \in \RR/2\pi\ZZ\}$.
    Using triangle inequality and the fact that~$\norm{A} \le 1$ and $\norm{C} \le 1$,
    we have
    \begin{align}
        \nonumber
        \norm[\big]{\avg_{R \sim \convo} \rho(R)}
        &=\norm*{
            \left(\avg_{U \sim \nu_{\mathsf{BRCC}}^{*m}} \rho(U) \right)
            \left(\avg_{\theta \sim \RR/2\pi\ZZ} \rho(e^{\ii \theta \sigma^z_1}) \right)
            \left(\avg_{V \sim \nu_{\mathsf{BRCC}}^{*m}} \rho(V) \right)
        }\\
        \nonumber&= \norm{CBC}\\
        &=\norm{ABA - AB(A-C) - (A-C)BC} \\
        \nonumber &\le
        \norm*{
            ABA
        }
        +
        2 \norm*{
            A-C
        }\\
        \nonumber
        &< 1 - \frac 1 {2^4} + \frac 2 {2^6} = 1 - \frac{1}{2^5}.
    \end{align}
    We apply~\ref{lem:detectability}
    to $L = 2m(n -1) + 1$ compact subgroups of~$\SU(2^n)$.
    Here, all but one subgroup~$H_0$
    are 2-qubit local Clifford groups 
    and there are $2m$ copies of~$\LCl_{i,i+1}$ for each $i = 1, 2, \ldots $, \mbox{$n-1$}.
    We conclude that for a convex combination
    \begin{equation}
        \Sigma = \frac{
            \mu(H_0) + 2m \sum_{i=1}^{n-1} \mu(\LCl_{i,i+1})
        }{1 + 2m (n - 1)} 
        \, ,
    \end{equation}
    we have
    \begin{equation}
        \norm*{
            \avg_{U \sim \Sigma} \rho(U)
        }
        \le 1 - \frac{1}{2^8 m n} \, .
    \end{equation}
    Since $\Pi(\rho|H') \preceq \Pi(\rho|H)$ for any $H' \ge H$,
    we enlarge the subgroups
    to the unitary groups on their support
    \begin{equation}
        \Sigma' = 
        \frac{
            \mu(\LU(\{1\})) + 2m \sum_{i=1}^{n-1} \mu(\LU(\{i,i+1\}))
        }{2m(n - 1) + 1}
    \end{equation}
    without increasing the norm~\cite[Lemma~2.21]{chen2025incompressibility} 
    \begin{align}
        \norm*{
            \avg_{V \sim \Sigma'} \rho(V)
        }
        \le
        \norm*{
            \avg_{U \sim \Sigma} \rho(U)
        }
        \le 1 - \frac{1}{2^8 m n} \, .
    \end{align}
    
    Next, we apply~\ref{lem:detectability} again.
    The subgroup $\LU(\{1\})$ has $2m$ other noncommuting subgroups,
    and each copy of~$\LU(\{i,i+1\})$ has at most
    $(2m-1)+ 2 \cdot 2m + 1 = 6m$ other noncommuting subgroups.
    So, we set $\ell = 6m$.
    A random unitary circuit of depth~$2m \cdot 2+1$ consisting of $\LU(1)$ and $\LU(2)$ gates
    now has a spectral gap lower bound
    $2^{-2} L \ell^{-2} \Delta(\Sigma') > (9 \cdot 2^{12} m^2)^{-1}$ .
    Here, all $U(4)$ gates are in a brickwork layout.
    The $\LU(1)$ gate projector can be absorbed into one of the $\LU(2)$ gate projectors,
    so we may ignore it.
    
    We conclude that the Brickwork Random Unitary Circuit of depth~$4m$,
    whose distribution is precisely~$\nu_{\mathsf{BRUC}}^{*2m}$,
    has a spectral gap lower bound
    \begin{equation}
        \Delta(\nu_{\mathsf{BRUC}}^{*2m}, \SU(2^n)) \ge \frac{2m(n-1) + 1}{9\cdot 2^{12} m^3 n} > \frac{1}{9\cdot 2^{12} m^2}
    \end{equation}
    for any~$m \ge 2^{10}$.

    We use~\ref{lem:detectability} one final time
    to conclude that
    \begin{equation}
        \norm*{
            \avg_{V \sim \Sigma''} \rho(V)
        }
        \le 1 - \frac{1}{9 \cdot 2^{45} (n-1)},
    \end{equation}
    where
    \begin{equation}
        \Sigma'' = \frac{
            2m \sum_{i=1}^{n-1} \mu(\LU_{i,i+1})
        }{2m(n - 1)}
        = 
        \frac{
            \sum_{i=1}^{n-1} \mu(\LU_{i,i+1})
        }{n - 1}. \qedhere
    \end{equation}
\end{proof}

\begin{theorem}[implying~\ref{thm:main}(3)]\label{thm:BRUC}
The spectral gap of the Brickwork Random Unitary Circuit of depth~$2$
(on the one-dimensional chain of $n$ qubits)
on $\SU(2^n)$ with $n \ge 2$ is
    \begin{align*}
        \Delta(\nu_{\mathsf{BRUC}},\SU(2^n)) > 2^{-53}.
    \end{align*}
\end{theorem}
\begin{proof}
    Apply~\ref{lem:detectability} with $L = n-1$ and $\ell=2$ to~\ref{eq:additive1DRandomUnitary}.
\end{proof}

The lower bound on the spectral gap of the Brickwork Random Unitary Circuit \ref{thm:BRUC}
transfers to the orthogonal group case.
\begin{theorem}\label{thm:BROC}
    The spectral gap of the Brickwork Random Orthogonal Circuit of depth~$2$
    (on the one-dimensional chain of $n$ qubits)
    on $\SO(2^n)$ with $n \ge 2$ is
    \begin{align*}
        \Delta(\nu_{\mathsf{BROC}},\SO(2^n)) > 2^{-56}.
    \end{align*}
\end{theorem}

\begin{proof}
    This exactly parallels the proof of \ref{thm:BRUC}.
    The analogous spectral gaps for the Random Pauli Orthogonal Rotation
    and the Brickwork Random Clifford Orthogonal Circuit
    are provided in \ref{thm:RPOR} and \ref{thm:BROCC}, respectively.
    The deformation of random unitary circuits
    uses the detectability lemma and its converse,
    which hold for compact subgroups of a compact group.
    The constant $m$ has to be adjusted to $m\ge 2^{11}$.
\end{proof}

\begin{remark}\label{rem:orthogonalMixingTime}
    The spectral gap of a distribution~$\nu$
    gives a bound on the 2-to-2 norm distance
    between the Haar random unitary/orthogonal channel
    and the $\nu$-random unitary/orthogonal channel
    acting on multiple copies of an $n$-qubit system.
    Converting this to a more quantum operationally meaningful metric
    such as the 1-to-1 norm or diamond distance,
    we obtain unitary/orthogonal designs.
    Specifically, the constant spectral gap result in~\ref{thm:BROC}    
    implies that a $1d$ Brickwork Random Orthogonal Circuit of depth~$\calO(nt)$
    is an $\epsilon$-approximate orthogonal $t$-design
    with a constant multiplicative or additive error~$\epsilon$~\cite{brandao2016local};
    see~\ref{lem:GapToDesign}.
    This implied mixing time with respect to the additive error is optimal for~$t=2$;
    it is shown~\cite[\S2.4]{schuster2025random} that
    any orthogonal circuit of depth~$o(n)$
    can be distinguished from Haar random orthogonal transformation
    by querying the ensemble on a two-copy system.
    See~\cite{haah2025short} for related discussion.
\end{remark}

\section{Unitary circuit complexity growth}
\label{sec:UnitaryComplexityGrowth}

\begin{definition}[Unitary circuit complexity]
    A \emph{(two-qubit) gate} is an element of the two-qubit local unitary group $\LU(\{a,b\},n)$,
    regarded as a unitary acting on two qubits $a$ and $b$.
    A \emph{circuit} is a finite sequence of two-qubit gates.
    By abuse of language,
    we say that a unitary~$U \in \SU(2^n)$ is a circuit if $U$ 
    is equal to the ordered product of the unitaries in a finite sequence of two-qubit gates,
    which is said to \emph{compile}~$U$.
    Given a unitary $U\in \SU(2^n)$,
    the number of gates in the smallest such circuit that (approximately) compiles $U$
    is called the \emph{unitary circuit complexity}:
    \begin{equation}
        C_{\delta}(U) = \min \{ m \mid \exists \text{ a circuit $V$ of $m$ gates such that } \norm{U-V} \le \delta\}
    \end{equation}
\end{definition}

\begin{definition}[Unitary design]
    Let $\nu$ be a probability distribution on $G=\SU(2^n)$.
    Consider the quantum channel corresponding to its moment operator in the balanced tensor representation $\tau_{t,t}$,
    \begin{equation}
        \Phi_\nu^{(t)}: A \mapsto \Ex_{U\sim \nu} U^{\otimes t} A U^{\dagger \otimes t}.
    \end{equation}
    We call $\nu$ an approximate unitary $t$-design with multiplicative error $\epsilon$ if 
    \begin{equation}
        (1-\epsilon) \Phi_{\mu(G)}^{(t)} \preceq \Phi_\nu^{(t)} \preceq (1+\epsilon) \Phi_{\mu(G)}^{(t)}
    \end{equation}
    where $\preceq$ denotes the completely positive ordering.
\end{definition}

\begin{lemma}[{\cite{brandao2016local},\cite[Lemma 2.6]{chen2025incompressibility}}]
\label{lem:GapToDesign}
    Suppose that the spectral gap of $\nu$ in the balanced tensor representation $\tau_{t,t}$ is at least $\Delta$.
    Then $\nu^{*k}$ is an approximate unitary $t$-design with multiplicative error $\epsilon$
    when $k \ge \frac{1}{\Delta} (nt \log 4 + \log \frac{1}{\epsilon})$.
\end{lemma}

\begin{corollary}
\label{cor:UnitaryComplexityGrowth}
    Let $\delta \in (0,\frac 1 2)$ and $\gamma \in (0,1)$ be constants
    and $U$ be a random circuit consisting of $L \le 4^{(1-\gamma) n}$ independent random samples 
    from the \emph{All-to-All Random Gate} ensemble
    \begin{equation}
        \nu_{\mathsf{All}} = \frac{2}{n(n-1)} \sum_{i < j} \mu(\LU(\{i,j\},n))\, .
    \end{equation}
    The unitary circuit complexity satisfies
    \begin{equation}
        C_{\delta}(U) = \Omega \left(\frac{L}{n \log (nL)}\right)
    \end{equation}
    with probability at least $1 - \exp(-\Omega(L/n))$ over the choice of~$U$.
    The big-$\Omega$ notation suppresses dependence on $\delta$.
\end{corollary}

This follows by replacing the fidelity of states 
with that of Choi states in the proof of \cite[Corollary 1.9]{chen2025incompressibility}.
One can enlarge the range of~$L$ by choosing $\gamma n = o(n)$
at the cost of a slightly weaker complexity lower bound.

\begin{proof}
    Let $\nu$ be an approximate unitary $t$-design with multiplicative error $\epsilon$,
    $\ket{\Phi}$ be the (normalized) maximally entangled state between a set of $n$ qubits and its copy,
    and $\delta \in (0,1)$.
    Two similar unitaries give high fidelity between the Choi states:
    \begin{equation}
        \norm{U - V} \le \delta \implies \abs*{\bra{\Phi} V^\dagger U \otimes I \ket{\Phi}}^2 
        \ge 1 - \delta^2
    \end{equation}
    We bound the probability of collision of $U \sim \nu$ with an arbitrary unitary~$V$,
    \begin{align}
        \nonumber
        \Pr_{U\sim \nu} \big[ \norm{U-V} \le \delta \big] 
        &\le \Pr_{U\sim \nu} \big[ \abs*{\bra{\Phi} UV^\dagger \otimes I \ket{\Phi}}^2 \ge 1 - \delta^2 \big]\\
        \nonumber
        &= \Pr_{U\sim \nu} \big[ \abs*{\bra{\Phi} UV^\dagger \otimes I \ket{\Phi}}^{2t} \ge (1 - \delta^2)^t \big]\\
        \nonumber
        &\le \frac{1}{(1-\delta^2)^t} \avg_{U\sim \nu} \abs*{\bra{\Phi} UV^\dagger \otimes I \ket{\Phi}}^{2t} && (\text{Markov's inequality})\\
        \label{eq:probEstimateForComplexity}
        &\le \frac{1+\epsilon}{(1-\delta^2)^t} \avg_{U\sim \mu(\SU(2^n))} \abs*{\bra{\Phi} UV^\dagger \otimes I \ket{\Phi}}^{2t} && (\text{approximate unitary $t$-design})\\
        \nonumber
        &= \frac{1+\epsilon}{(1-\delta^2)^t} \avg_{U\sim \mu(\SU(2^n))} \abs*{\frac{\Tr(U)}{2^n}}^{2t}
        && (\text{Haar invariance})\\
        \nonumber
        &= \frac{1+\epsilon}{(1-\delta^2)^t} \cdot \frac{\Tr (\avg_{U\sim \mu(\SU(2^n))} \tau_{t,t}(U))}{4^{nt}}\\
        \nonumber
        &\le \frac{1+\epsilon}{(1-\delta^2)^t} \cdot \frac{t!}{4^{nt}}
    \end{align}
    where the last line uses the fact 
    that the moment operator $\avg_{U\sim \mu(\SU(2^n))} \tau_{t,t}(U)$ is the orthogonal projector onto the $t$-copy commutant of $\SU(2^n)$.
    By Schur--Weyl duality, 
    this commutant is spanned by permutation operators on $t$ tensor factors.
    When $2^n \ge t$, these permutation operators are linearly independent;
    otherwise they are linearly dependent.
    In either case, the dimension of the commutant is at most $t!$, 
    so the trace of the projector is at most $t!$.
    See e.g. \cite[Equation 113]{mele2024introduction}.
    
    We construct a covering net on the set of shallow circuits
    to bound their probability of collision with $\nu$.
    A $\delta_2$-net over a metric space $G$ 
    is a subset $S \subseteq G$ such that for every $V\in G$, 
    there is a $V'\in S$ such that $\norm{V-V'}\le \delta_2$.
    For $\LU(\{a,b\},n)$ with the operator norm, a finite $\delta_2$-net $S(\delta_2)$ 
    is well known to exist with $(\calO(1)/\delta_2)^{16}$ elements.
    See e.g.~\cite{milman1986asymptotic}.
    Then, the set of all circuits with at most~$R$ two-qubit gates
    admits a finite $R\delta_2$-net, denoted by~$M(\delta_2,R)$,
    using $\left({n\choose 2} \abs{S(\delta_2)}\right)^R$ elements
    by replacing every gate in a circuit by an element from~$S(\delta_2)$.
    
    If a circuit with~$R$ gates happens to be close to a unitary~$U$,
    then there is an element in the net $M(\delta_2,R)$ that is also close to~$U$: 
    \begin{equation}
    \begin{split}
        & \text{$\exists$ circuit $V$ using $R$ gates such that }
        \abs*{\bra{\Phi} VU^\dagger \otimes I \ket{\Phi}}^2 \ge 1 - \delta^2\\
        \implies
        &\text{$\exists\; V'\in M(\delta_2,R)$ such that } \abs*{\bra{\Phi} V'U^\dagger \otimes I \ket{\Phi}}^2 \ge 1 - \delta^2 - 2R\delta_2.
    \end{split}
    \end{equation}
    Let $\delta_2 = \frac{\delta^2}{2R}$.
    The probability of collision between shallow circuits and $\nu$ is
    \begin{align}
        \nonumber
        &\Pr_{U\sim \nu} \left[\text{$\exists$ circuit $V$ using $R$ gates such that } \abs*{\bra{\Phi} V U^\dagger \otimes I \ket{\Phi}}^2 \ge 1 - \delta^2 \right]\\
        \nonumber
        &\le \Pr_{U\sim \nu} \left[\text{$\exists\; V' \in M(\delta_2,R)$ such that } \abs*{\bra{\Phi} V' U^\dagger \otimes I \ket{\Phi}}^2 \ge 1 - 2\delta^2 \right]\\
        &\le \abs[\big]{M(\delta_2,R)} \cdot \Pr_{U\sim \nu} [ \abs*{\bra{\Phi} V' U^\dagger \otimes I \ket{\Phi}}^2 \ge 1 - 2\delta^2 ]\\
        \nonumber
        &\le \abs[\big]{M(\delta_2,R)} \cdot \frac{1+\epsilon}{(1-2\delta^2)^t} \cdot \frac{t!}{4^{nt}}&&(\text{by }\ref{eq:probEstimateForComplexity})\\
        \nonumber
        &\le \left(n^2 \left(\frac{\calO(1)R}{\delta^2} \right)^{16}\right)^R \cdot \frac{1+\epsilon}{(1-2\delta^2)^t} \cdot \frac{t!}{4^{nt}}.
    \end{align}
    We realize the unitary design $\nu$ as the probability distribution induced by a deep random circuit.
    From \ref{thm:AdditiveRandomUnitary},
    the spectral gap of the All-to-All Random Gate $\nu_\mathsf{All}$ is $\Theta(1/n)$.
    This implies that $\nu = \nu_\mathsf{All}^{*L}$ is an approximate unitary $t$-design with multiplicative error $\epsilon$
    when $L = \Omega(n^2 t + n\log \frac{1}{\epsilon})$ by \ref{lem:GapToDesign}.
    This can be rephrased as a random circuit using $L$ gates is an approximate unitary $t$-design
    with a constant multiplicative error, say $\epsilon = 0.5$, for $t = \calO(L/n^2)$.
    From the assumption $L \le 4^{(1-\gamma)n}$, then $t\le 4^{(1-\gamma)n}$ implies
    \begin{equation}
        \log \left(\frac{t!}{4^{nt}}\right) \le t\log t - nt\log 4 \le - \gamma nt\log 4.
    \end{equation}
    Choosing $R= c\cdot nt/\log(nt)$ where $c>0$ is a constant,
    \begin{equation}
        \log \left(n^{2R} R^{16R}\right) = 2R\log n + 16R\log R \le c\cdot 18nt.
    \end{equation}
    The probability of collision between shallow circuits using $R$ gates and deep circuits using $L$ gates is exponentially small as long as $18 c < \gamma \log 4$,
    \begin{equation}
        \Pr_{U\sim \nu} \left[
            \text{$\exists$ circuit $V$ using $R$ gates such that } \norm{U-V} \le \delta
        \right]
        \le e^{-\Omega(L/n)}. 
    \end{equation}
    We conclude that circuits using $L$ gates have unitary circuit complexity $R=\Omega(L/(n\log(L/n)))$ with probability at least $1 - \exp(-\Omega(L/n))$.
    The complexity growth saturates approaching $t=4^n$ since the contribution $\log(t!/4^{nt})$ is no longer sufficiently negative.
\end{proof}

{\bf Acknowledgments.}
This material is based upon work supported by the U.S. Department of
Energy, Office of Science, Office of Advanced Scientific Computing Research, Department of
Energy Computational Science Graduate Fellowship under Award Number DE-SC0026073 (T.B.).
T.B. and J.H. are in part supported by DOE grant DE-SC0025934.
We have used AI for finite calculation as noted in~\ref{rem:gapCalculationAI} and the spectral gap upper bound as noted in~\ref{rem:gapUpperBoundAIProof}.
We have also used AI for transcription of Figure~\ref{fig:deformation},
literature search,
and proofreading.
\\

{\bf Disclaimer.}
This report was prepared as an account of work sponsored by an agency of the
United States Government. Neither the United States Government nor any agency thereof, nor
any of their employees, makes any warranty, express or implied, or assumes any legal liability
or responsibility for the accuracy, completeness, or usefulness of any information, apparatus,
product, or process disclosed, or represents that its use would not infringe privately owned
rights. Reference herein to any specific commercial product, process, or service by trade name,
trademark, manufacturer, or otherwise does not necessarily constitute or imply its
endorsement, recommendation, or favoring by the United States Government or any agency
thereof. The views and opinions of authors expressed herein do not necessarily state or reflect
those of the United States Government or any agency thereof.

\bibliographystyle{alphaurl}
\bibliography{walkrefs}

\newcommand{\etalchar}[1]{$^{#1}$}
\begin{thebibliography}{HMMH{\etalchar{+}}23}

\bibitem[AALV09]{aharonov2009detectability}
Dorit Aharonov, Itai Arad, Zeph Landau, and Umesh Vazirani.
\newblock The detectability lemma and quantum gap amplification.
\newblock In {\em Proceedings of the forty-first annual ACM symposium on Theory of computing}, pages 417--426, 2009.
\newblock \href {https://arxiv.org/abs/0811.3412} {\path{arXiv:0811.3412}}, \href {https://doi.org/10.1145/1536414.1536472} {\path{doi:10.1145/1536414.1536472}}.

\bibitem[AAV16]{anshu2016simple}
Anurag Anshu, Itai Arad, and Thomas Vidick.
\newblock Simple proof of the detectability lemma and spectral gap amplification.
\newblock {\em Physical Review B}, 93(20):205142, 2016.

\bibitem[ABH{\etalchar{+}}26]{anshu2026depth1expandersunitarygroup}
Anurag Anshu, Shankar Balasubramanian, Jonas Haferkamp, Aram~W. Harrow, and Xinyu Tan.
\newblock Depth-1 expanders on the unitary group and applications, 2026.
\newblock \href {https://arxiv.org/abs/2609.01605} {\path{arXiv:2609.01605}}.

\bibitem[AKLT88]{affleck1988valence}
Ian Affleck, Tom Kennedy, Elliott~H Lieb, and Hal Tasaki.
\newblock Valence bond ground states in isotropic quantum antiferromagnets.
\newblock {\em Communications in Mathematical Physics}, 115(3):477--528, 1988.

\bibitem[Ans20]{anshu2020improved}
Anurag Anshu.
\newblock Improved local spectral gap thresholds for lattices of finite size.
\newblock {\em Physical Review B}, 101(16):165104, 2020.

\bibitem[AP26]{alon2026aldous}
Gil Alon and Doron Puder.
\newblock Aldous-type spectral gaps in unitary groups.
\newblock {\em arXiv preprint arXiv:2603.00353}, 2026.

\bibitem[BCHJ{\etalchar{+}}21]{brandao2021models}
Fernando~GSL Brand{\~a}o, Wissam Chemissany, Nicholas Hunter-Jones, Richard Kueng, and John Preskill.
\newblock Models of quantum complexity growth.
\newblock {\em PRX Quantum}, 2(3):030316, 2021.

\bibitem[BEL{\etalchar{+}}25]{bittel2025complete}
Lennart Bittel, Jens Eisert, Lorenzo Leone, Antonio~A Mele, and Salvatore~FE Oliviero.
\newblock A complete theory of the clifford commutant.
\newblock {\em arXiv preprint arXiv:2504.12263}, 2025.

\bibitem[BG12]{bourgain2012spectral}
Jean Bourgain and Alex Gamburd.
\newblock A spectral gap theorem in $su(d)$.
\newblock {\em Journal of the European Mathematical Society}, 14(5):1455--1511, 2012.
\newblock \href {https://arxiv.org/abs/1108.6264} {\path{arXiv:1108.6264}}, \href {https://doi.org/10.4171/JEMS/337} {\path{doi:10.4171/JEMS/337}}.

\bibitem[BHH16]{brandao2016local}
Fernando~GSL Brand{\~a}o, Aram~W Harrow, and Micha{\l} Horodecki.
\newblock Local random quantum circuits are approximate polynomial-designs.
\newblock {\em Communications in Mathematical Physics}, 346:397--434, 2016.
\newblock \href {https://arxiv.org/abs/1208.0692} {\path{arXiv:1208.0692}}, \href {https://doi.org/10.1007/s00220-016-2706-8} {\path{doi:10.1007/s00220-016-2706-8}}.

\bibitem[BS18]{brown2018second}
Adam~R Brown and Leonard Susskind.
\newblock Second law of quantum complexity.
\newblock {\em Physical Review D}, 97(8):086015, 2018.

\bibitem[BtD85]{BrockertomDieck}
Theodore Br\"ocker and Tammo tom Dieck.
\newblock {\em Representations of compact {Lie} groups}.
\newblock Graduate Texts in Mathematics. Springer, 1985.

\bibitem[Bum04]{bump2004lie}
Daniel Bump.
\newblock {\em Lie groups}, volume 225.
\newblock Springer, 2004.

\bibitem[BV10]{brown2010convergence}
Winton~G Brown and Lorenza Viola.
\newblock Convergence rates for arbitrary statistical moments of random quantum circuits.
\newblock {\em Physical review letters}, 104(25):250501, 2010.
\newblock \href {https://arxiv.org/abs/0910.0913} {\path{arXiv:0910.0913}}, \href {https://doi.org/10.1103/PhysRevLett.104.250501} {\path{doi:10.1103/PhysRevLett.104.250501}}.

\bibitem[Cap08]{caputo2008spectral}
Pietro Caputo.
\newblock On the spectral gap of the kac walk and other binary collision processes.
\newblock {\em arXiv preprint arXiv:0807.3415}, 2008.

\bibitem[CCL03]{carlen2003determination}
EA~Carlen, MC~Carvalho, and M~Loss.
\newblock Determination of the spectral gap for kac's master equation and related stochastic evolution.
\newblock {\em Acta Mathematica}, 191(1):1--54, 2003.

\bibitem[CDX{\etalchar{+}}24]{chen2024efficient}
Chi-Fang Chen, Jordan Docter, Michelle Xu, Adam Bouland, Fernando~GSL Brand{\~a}o, and Patrick Hayden.
\newblock Efficient unitary designs from random sums and permutations.
\newblock In {\em 2024 IEEE 65th Annual Symposium on Foundations of Computer Science (FOCS)}, pages 476--484. IEEE, 2024.

\bibitem[CHH{\etalchar{+}}25]{chen2025incompressibility}
Chi-Fang Chen, Jeongwan Haah, Jonas Haferkamp, Yunchao Liu, Tony Metger, and Xinyu Tan.
\newblock Incompressibility and spectral gaps of random circuits.
\newblock In {\em 2025 IEEE 66th Annual Symposium on Foundations of Computer Science (FOCS)}, pages 1304--1312. IEEE, 2025.

\bibitem[Chr06]{christandl2006structure}
Matthias Christandl.
\newblock The structure of bipartite quantum states-insights from group theory and cryptography.
\newblock {\em Ph. D. Thesis}, 2006.

\bibitem[CL26]{chen2026spectralgapdownupwalks}
Xiaoyu Chen and Kuikui Liu.
\newblock Spectral gap of down-up walks via trickle-down: A simplified and sharpened analysis, 2026.
\newblock \href {https://arxiv.org/abs/2609.19514} {\path{arXiv:2609.19514}}.

\bibitem[CLR10]{caputo2010proof}
Pietro Caputo, Thomas Liggett, and Thomas Richthammer.
\newblock Proof of {Aldous}’ spectral gap conjecture.
\newblock {\em Journal of the American Mathematical Society}, 23(3):831--851, 2010.

\bibitem[DF26]{dong2026spectral}
Ziyuan Dong and Xiang Fan.
\newblock Spectral gaps and slow modes of pauli rotations and random quantum circuits, 2026.
\newblock \href {https://arxiv.org/abs/2609.40164} {\path{arXiv:2609.40164}}.

\bibitem[Dia88]{diaconis1988group}
Persi Diaconis.
\newblock {\em Group representations in probability and statistics}, volume~11.
\newblock Institute of mathematical statistics Hayward, CA, 1988.

\bibitem[DLT02]{DiVincenzo_2002}
D.P. DiVincenzo, D.W. Leung, and B.M. Terhal.
\newblock Quantum data hiding.
\newblock {\em IEEE Transactions on Information Theory}, 48(3):580–598, March 2002.
\newblock \href {https://arxiv.org/abs/quant-ph/0103098} {\path{arXiv:quant-ph/0103098}}, \href {https://doi.org/10.1109/18.985948} {\path{doi:10.1109/18.985948}}.

\bibitem[DS81]{diaconis1981generating}
Persi Diaconis and Mehrdad Shahshahani.
\newblock Generating a random permutation with random transpositions.
\newblock {\em Zeitschrift f{\"u}r Wahrscheinlichkeitstheorie und verwandte Gebiete}, 57(2):159--179, 1981.

\bibitem[DSC93]{DiaconisSaloffCoste1993}
Persi Diaconis and Laurent Saloff-Coste.
\newblock Comparison theorems for reversible markov chains.
\newblock {\em Ann. Appl. Probab.}, 3(3):696--730, 1993.
\newblock \href {https://doi.org/10.1214/aoap/1177005359} {\path{doi:10.1214/aoap/1177005359}}.

\bibitem[FH13]{fulton2013representation}
William Fulton and Joe Harris.
\newblock {\em Representation theory: a first course}.
\newblock Springer Science \& Business Media, 2013.

\bibitem[FKNV23]{Fisher_2023}
Matthew~P.A. Fisher, Vedika Khemani, Adam Nahum, and Sagar Vijay.
\newblock Random quantum circuits.
\newblock {\em Annual Review of Condensed Matter Physics}, 14(1):335–379, March 2023.
\newblock URL: \url{http://dx.doi.org/10.1146/annurev-conmatphys-031720-030658}, \href {https://doi.org/10.1146/annurev-conmatphys-031720-030658} {\path{doi:10.1146/annurev-conmatphys-031720-030658}}.

\bibitem[Gao15]{gao2015quantum}
Jingliang Gao.
\newblock Quantum union bounds for sequential projective measurements.
\newblock {\em Physical Review A}, 92(5):052331, 2015.
\newblock \href {https://arxiv.org/abs/1410.5688} {\path{arXiv:1410.5688}}, \href {https://doi.org/10.1103/PhysRevA.92.052331} {\path{doi:10.1103/PhysRevA.92.052331}}.

\bibitem[GM16]{gosset2016local}
David Gosset and Evgeny Mozgunov.
\newblock Local gap threshold for frustration-free spin systems.
\newblock {\em Journal of Mathematical Physics}, 57(9), 2016.

\bibitem[Gro26]{GAP4}
The~GAP Group.
\newblock {GAP -- Groups, Algorithms, and Programming, Version 4.16.0}, 2026.
\newblock URL: \url{https://www.gap-system.org}.

\bibitem[Haa17]{Haah_2017}
Jeongwan Haah.
\newblock Algebraic methods for quantum codes on lattices.
\newblock {\em Revista Colombiana de Matemáticas}, 50(2):299, January 2017.
\newblock URL: \url{http://dx.doi.org/10.15446/recolma.v50n2.62214}, \href {https://doi.org/10.15446/recolma.v50n2.62214} {\path{doi:10.15446/recolma.v50n2.62214}}.

\bibitem[Haa25a]{haah2025short}
Jeongwan Haah.
\newblock Short remarks on shallow unitary circuits.
\newblock {\em Quantum}, 9:1940, 2025.

\bibitem[Haa25b]{haah2025topological}
Jeongwan Haah.
\newblock Topological phases of unitary dynamics: Classification in {Clifford} category.
\newblock {\em Communications in Mathematical Physics}, 406(4):76, 2025.

\bibitem[Haf22]{haferkamp2022random}
Jonas Haferkamp.
\newblock Random quantum circuits are approximate unitary $t$-designs in depth ${O}(nt^{5+o(1)})$.
\newblock {\em Quantum}, 6:795, 2022.

\bibitem[HdL16]{haagerup2016simple}
Uffe Haagerup and Tim de~Laat.
\newblock Simple lie groups without the approximation property ii.
\newblock {\em Transactions of the American Mathematical Society}, 368(6):3777--3809, 2016.

\bibitem[HHJ21]{haferkamp2021improved}
Jonas Haferkamp and Nicholas Hunter-Jones.
\newblock Improved spectral gaps for random quantum circuits: Large local dimensions and all-to-all interactions.
\newblock {\em Physical Review A}, 104(2):022417, 2021.

\bibitem[HJ19]{hunter2019unitary}
Nicholas Hunter-Jones.
\newblock Unitary designs from statistical mechanics in random quantum circuits.
\newblock {\em arXiv preprint arXiv:1905.12053}, 2019.

\bibitem[HJL25]{hunter2025two}
Nicholas Hunter-Jones and Marius Lemm.
\newblock Two classes of quantum spin systems that are gapped on any bounded-degree graph.
\newblock {\em arXiv preprint arXiv:2509.22438}, 2025.

\bibitem[HLT25]{haah2025efficient}
Jeongwan Haah, Yunchao Liu, and Xinyu Tan.
\newblock Efficient approximate unitary designs from random pauli rotations.
\newblock {\em Communications in Mathematical Physics}, 406(12):1--24, 2025.

\bibitem[HMMH{\etalchar{+}}23]{haferkamp2023efficient}
Jonas Haferkamp, Felipe Montealegre-Mora, Markus Heinrich, Jens Eisert, David Gross, and Ingo Roth.
\newblock Efficient unitary designs with a system-size independent number of non-clifford gates.
\newblock {\em Communications in Mathematical Physics}, 397(3):995--1041, 2023.

\bibitem[Isa76]{Isaacs1976}
I.~Martin Isaacs.
\newblock {\em Character Theory of Finite Groups}.
\newblock AMS Chelsea Publishing, 1976.

\bibitem[JBS23]{Jian2023}
Shao-Kai Jian, Gregory Bentsen, and Brian Swingle.
\newblock Linear growth of circuit complexity from brownian dynamics.
\newblock {\em Journal of High Energy Physics}, 2023(8):190, Aug 2023.
\newblock \href {https://arxiv.org/abs/2206.14205} {\path{arXiv:2206.14205}}, \href {https://doi.org/10.1007/JHEP08(2023)190} {\path{doi:10.1007/JHEP08(2023)190}}.

\bibitem[Kac56]{kac1956foundations}
M~Kac.
\newblock Foundations of kinetic theory.
\newblock In {\em Proc. Third Berkely Symp. on Math. Stat. and Prob}, volume~3, pages 171--197, 1956.

\bibitem[Kna88]{knabe1988energy}
Stefan Knabe.
\newblock Energy gaps and elementary excitations for certain vbs-quantum antiferromagnets.
\newblock {\em Journal of statistical physics}, 52(3):627--638, 1988.

\bibitem[Koi89]{koike1989decomposition}
Kazuhiko Koike.
\newblock On the decomposition of tensor products of the representations of the classical groups: by means of the universal characters.
\newblock {\em Advances in Mathematics}, 74(1):57--86, 1989.

\bibitem[Lem19]{lemm2019finite}
Marius Lemm.
\newblock Finite-size criteria for spectral gaps in $ d $-dimensional quantum spin systems.
\newblock {\em arXiv preprint arXiv:1902.07141}, 2019.

\bibitem[LH26]{liu2026almostquadruplyoptimalunitary}
Guoding Liu and Jonas Helsen.
\newblock (almost) quadruply optimal unitary designs in 1d, 2026.
\newblock \href {https://arxiv.org/abs/2608.18650} {\path{arXiv:2608.18650}}.

\bibitem[Li26]{li2026gatetimecomplexitygrowth}
Zhi Li.
\newblock One gate at a time: Complexity growth in random quantum circuits, 2026.
\newblock \href {https://arxiv.org/abs/2609.17457} {\path{arXiv:2609.17457}}.

\bibitem[LL26]{laracuente2026approximate}
Nicholas LaRacuente and Felix Leditzky.
\newblock Approximate unitary k-designs from shallow, low-communication circuits.
\newblock {\em Communications in Mathematical Physics}, 407(3):51, 2026.

\bibitem[LM19]{lemm2019spectral}
Marius Lemm and Evgeny Mozgunov.
\newblock Spectral gaps of frustration-free spin systems with boundary.
\newblock {\em Journal of Mathematical Physics}, 60(5), 2019.

\bibitem[LX22]{lemm2022quantitatively}
Marius Lemm and David Xiang.
\newblock Quantitatively improved finite-size criteria for spectral gaps.
\newblock {\em Journal of Physics A: Mathematical and Theoretical}, 55(29):295203, 2022.

\bibitem[Mas03]{maslen2003eigenvalues}
David~K Maslen.
\newblock The eigenvalues of kac's master equation.
\newblock {\em Mathematische Zeitschrift}, 243(2):291--331, 2003.

\bibitem[Mel24]{mele2024introduction}
Antonio~Anna Mele.
\newblock Introduction to haar measure tools in quantum information: A beginner's tutorial.
\newblock {\em Quantum}, 8:1340, 2024.

\bibitem[MHJ23]{mittal2023local}
Shivan Mittal and Nicholas Hunter-Jones.
\newblock Local random quantum circuits form approximate designs on arbitrary architectures.
\newblock {\em arXiv preprint arXiv:2310.19355}, 2023.

\bibitem[MPSY24]{metger2024simple}
Tony Metger, Alexander Poremba, Makrand Sinha, and Henry Yuen.
\newblock Simple constructions of linear-depth t-designs and pseudorandom unitaries.
\newblock In {\em 2024 IEEE 65th Annual Symposium on Foundations of Computer Science (FOCS)}, pages 485--492. IEEE, 2024.

\bibitem[MS86]{milman1986asymptotic}
Vitali~D Milman and Gideon Schechtman.
\newblock {\em Asymptotic theory of finite dimensional normed spaces}.
\newblock Springer, 1986.

\bibitem[NRS01]{nebe2000invariantscliffordgroups}
Gabriele Nebe, E.~M. Rains, and N.~J.~A. Sloane.
\newblock The invariants of the clifford groups, 2001.
\newblock \href {https://arxiv.org/abs/math/0001038} {\path{arXiv:math/0001038}}, \href {https://doi.org/10.1023/A:1011233615437} {\path{doi:10.1023/A:1011233615437}}.

\bibitem[NRVH17]{Nahum_2017}
Adam Nahum, Jonathan Ruhman, Sagar Vijay, and Jeongwan Haah.
\newblock Quantum entanglement growth under random unitary dynamics.
\newblock {\em Physical Review X}, 7(3), 2017.
\newblock \href {https://arxiv.org/abs/1608.06950} {\path{arXiv:1608.06950}}, \href {https://doi.org/10.1103/physrevx.7.031016} {\path{doi:10.1103/physrevx.7.031016}}.

\bibitem[NVH18]{Nahum_2018}
Adam Nahum, Sagar Vijay, and Jeongwan Haah.
\newblock Operator spreading in random unitary circuits.
\newblock {\em Physical Review X}, 8(2), 2018.
\newblock \href {https://arxiv.org/abs/1705.08975} {\path{arXiv:1705.08975}}, \href {https://doi.org/10.1103/physrevx.8.021014} {\path{doi:10.1103/physrevx.8.021014}}.

\bibitem[Oli09]{Oliveira_2009}
Roberto~Imbuzeiro Oliveira.
\newblock On the convergence to equilibrium of kac’s random walk on matrices.
\newblock {\em The Annals of Applied Probability}, 19(3), 2009.
\newblock \href {https://arxiv.org/abs/0705.2253} {\path{arXiv:0705.2253}}, \href {https://doi.org/10.1214/08-aap550} {\path{doi:10.1214/08-aap550}}.

\bibitem[OV22]{o2022quantum}
Ryan O'Donnell and Ramgopal Venkateswaran.
\newblock The quantum union bound made easy.
\newblock In {\em Symposium on Simplicity in Algorithms (SOSA)}, pages 314--320. SIAM, 2022.

\bibitem[Sch90]{SCHNEIDER1990601}
Gerhard~J.A. Schneider.
\newblock Dixon's character table algorithm revisited.
\newblock {\em Journal of Symbolic Computation}, 9(5):601--606, 1990.
\newblock URL: \url{https://www.sciencedirect.com/science/article/pii/S0747717108800776}, \href {https://doi.org/10.1016/S0747-7171(08)80077-6} {\path{doi:10.1016/S0747-7171(08)80077-6}}.

\bibitem[SHH]{schuster2025random}
Thomas Schuster, Jonas Haferkamp, and Hsin-Yuan Huang.
\newblock Random unitaries in extremely low depth.
\newblock \href {https://arxiv.org/abs/2407.07754} {\path{arXiv:2407.07754}}.

\bibitem[WWT{\etalchar{+}}]{atlas}
Robert Wilson, Peter Walsh, Jonathan Tripp, Ibrahim Suleiman, Richard Parker, Simon Norton, Simon Nickerson, Steve Linton, John Bray, , and Rachel Abbott.
\newblock Atlas of finite group representations.
\newblock URL: \url{https://brauer.maths.qmul.ac.uk/Atlas/v3/group/S62/}.

\bibitem[ZYHC20]{Zhou_2020}
Shiyu Zhou, Zhicheng Yang, Alioscia Hamma, and Claudio Chamon.
\newblock Single t gate in a clifford circuit drives transition to universal entanglement spectrum statistics.
\newblock {\em SciPost Physics}, 9(6), December 2020.
\newblock \href {https://arxiv.org/abs/1906.01079} {\path{arXiv:1906.01079}}, \href {https://doi.org/10.21468/scipostphys.9.6.087} {\path{doi:10.21468/scipostphys.9.6.087}}.

\end{thebibliography}

\end{document}